%% file: main.tex
\documentclass[sigconf,nonacm, screen]{acmart}
\usepackage{graphicx} 
\usepackage{algorithm}
\usepackage{algpseudocode}
\usepackage{bbm}
\usepackage{multirow}
\usepackage{soul}
\usepackage{bm}
\usepackage{slnotation}
\usepackage{booktabs}

\usepackage{enumitem}
\usepackage{subcaption}
\input{XX-notation}

\newcommand{\capicon}[1]{%
  \raisebox{-0.2ex}{\includegraphics[height= 2.0 ex]{#1}}}

\DeclareMathOperator*{\argmin}{arg\,min}
\newcommand{\area}{\capicon{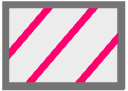}}

\newtheorem{theorem}{Theorem}
\newtheorem{proposition}{Proposition}
\newtheorem{example}{Example}

\author{Rikiya Takehi}
\authornote{Preprint of a paper accepted at WSDM 2026. Final version soon to appear.}
\orcid{0009-0003-6336-8064}
\affiliation{
  \institution{Waseda University}
  \city{Tokyo}
  \country{Japan}
}
\email{rikiya.takehi@fuji.waseda.jp}

\author{Fernando Diaz}
\orcid{0000-0003-2345-1288}
\affiliation{%
  \institution{Carnegie Mellon University}
  \city{Pittsburgh}
  \state{PA}
  \country{USA}
}
\email{diazf@acm.org}

\author{Tetsuya Sakai}
\orcid{0000-0002-6720-963X}
\affiliation{
  \institution{Waseda University}
  \city{Tokyo}
  \country{Japan}
}
\email{tetsuya@waseda.jp}

\title{Diversification as Risk Minimization}

\begin{document}
\begin{abstract}
Users tend to remember failures of a search session more than its many successes. This observation has led to work on search robustness, where systems are penalized if they perform very poorly on some queries. However, this principle of robustness has been overlooked within a single query. An ambiguous or underspecified query (e.g., ``jaguar'') can have several user intents, where popular intents often dominate the ranking, leaving users with minority intents unsatisfied. Although the diversification literature has long recognized this issue, existing metrics only model the \emph{average} relevance across intents and provide no robustness guarantees. More surprisingly, we show theoretically and empirically that many well‑known diversification algorithms are no more robust than a naive, non-diversified algorithm. To address this critical gap, we propose to frame diversification as a risk-minimization problem. We introduce \textbf{VRisk}, which measures the expected risk faced by the least‑served fraction of intents in a query. Optimizing VRisk produces a robust ranking, reducing the likelihood of poor user experiences. We then propose \textbf{VRisker}, a fast greedy re‑ranker with provable approximation guarantees. Finally, experiments on NTCIR INTENT‑2, TREC Web 2012, and MovieLens show the vulnerability of existing methods. VRisker reduces worst-case intent failures by up to 33\% with a minimal 2\% drop in average performance.

Code provided in \url{https://github.com/RikiyaT/VRisk}.
\end{abstract}

\maketitle

\section{Introduction}
Users remember the worst search sessions far more than the average ones~\citep{li2009worstintent, khabsa2016worstintent, white2009worstintent, collins2009risk, wang2012risk}. This insight has motivated extensive research on search robustness~\citep{wang2012risk,dincer2016risk, Diaz_2024, voorhees2005robusttrec, dincer2014risk, collins2009risk, robertson2006risk}. These prior works typically consider the downside risk, with the goal of minimizing the chances that users are unsatisfied with the search results. The fundamental principle is that search systems that perform poorly on some queries should be penalized, even if they perform well on average. However, these studies have only focused on the risks of rankings at the \textit{query level} and have not considered the potential downside risks \textit{within a query}.

This limitation becomes particularly problematic when queries are ambiguous or underspecified, leading to multiple possible user intents~\citep{clarke2009diversity}. Consider a search for ``jaguar.'' While many users may seek information about the car brand, those interested in the animal may find search results entirely dominated by the car brand. Naive ranking algorithms, designed to maximize a single relevance score, naturally produce such unbalanced results, where users with minority intents end up failing their search sessions.

The intent-based search diversification literature has long attempted to address this robustness-within-query challenge, aiming to generate a ranking to cover multiple intents~\citep{agrawal2009diversifysearch, radlanski2008diverseonline, sakai2011diversity, clarke2008alphadiversity, pm2dang2012}. Yet, although preventing intent failures is a core motivation for diversity, existing approaches typically measure the average relevance across intents rather than explicitly addressing the worst‑case outcomes. Importantly, we show mathematically that, in many cases, the family of intent-weighted diversity metrics~\citep{clarke2008alphadiversity, agrawal2009diversifysearch, sakai2011diversity} offers \textit{no protection} in terms of robustness. For example, in the toy example of Table~\ref{tab:combined}, we show that NDCG-IA~\citep{agrawal2009diversifysearch} favors a ranking that completely ignores 49\% of user intents. Consequently, as our experiments reveal, many existing diversification algorithms are no more robust than a naive, non-diversified ranking.

In light of this, we convert diversification into a principled risk-minimization problem. We argue that a search result must be robust within a query, in a way that minimizes the chances that users are unsatisfied. Specifically, we introduce the first framework for intent-aware risk minimization within a ranking, providing an evaluation metric \textbf{VRisk}. Drawing inspiration from finance, VRisk analyzes risk by adapting Conditional Value at Risk (CVaR), evaluating the expected relevance loss for the least-addressed fraction of user intents within each query. VRisk answers the question: ``\textit{How bad is the ranking for the worst $\beta$-fraction of intents.}'' VRisk can be baseline-free, measuring absolute risk, but can also be computed relative to a baseline system, following the common practice in existing IR risk metrics like URisk~\citep{wang2012risk} and GeoRisk~\citep{dincer2016risk}. VRisk is tunable and intuitive, giving both risk and guarantee in its evaluation. Additionally, we offer an efficient optimization algorithm called \textbf{VRisker} with a strong optimality guarantee, making our approach practical and theoretically robust for real-world problems.

Finally, we demonstrate that the robustness problem is mitigated through experiments on NTCIR INTENT-2~\citep{sakai2013ntcirintenttask} and TREC Web 2012~\citep{Clarke2012WebTrack} datasets. Additionally, this problem extends beyond search, as intent-based diversity has become a key concern in modern recommender systems as well~\citep{Wang2025mood, Vargas2011Intent,Jannach2024Survey,Li2023IntEL,Chang2023Latent,Sun2024LLM,Wang2023IDCL}. Thus, we also test on the MovieLens 32M~\citep{harper2015movielenslatest} to verify that our risk-aware framework addresses the same issues in recommender systems. Our results show that the algorithm reduces worst-case intent failures by up to 33\% while sacrificing as little as 2\% in average performance.

Our contributions are summarized as follows:
\begin{itemize}
    \item We present the first framework for intent-aware risk minimization in search, introducing a novel CVaR-based metric for both absolute and baseline-relative risk assessment, formally bridging diversification and robustness.
    \item We demonstrate mathematically and empirically that many mainstream diversification methods are fundamentally not robust, often performing no better than a naive baseline.
    \item We develop an efficient greedy algorithm with a strong approximation guarantee.
    \item We demonstrate through extensive experiments on three datasets that our algorithm consistently reduces worst-case intent failures by up to 33\% at the cost of only a 2\% drop in average relevance.
\end{itemize}

\section{Related Work}
\label{sec: related work}
\subsection{Risk-Aware Search and Recommendation}
Most influential works on retrieval robustness treat the entire query, rather than individual intents, as the unit of risk~\citep{wang2012risk, dincer2014risk, collins2009risk, robertson2006risk, wen2022drorecrisk}. The proposed risk-aware frameworks generally penalize systems that perform poorly on some queries, instead of looking at the average performance. By favoring such robust algorithms, they minimize the probability of users being unsatisfied. However, while considering risks at the query level, they have not considered the risk that intents may be poorly addressed within a single query.

Portfolio-theory-based approaches~\citep{tsai2012pmpt, wang2009portfolio} assess the variance of a single ranking by treating each document’s relevance score as the ambiguous factor. This perspective differs fundamentally from ours, as we focus on ambiguity at the intent level rather than at the document level.

Other studies consider risks in decision-making, search, and recommendation~\citep{huang2021offpolicyriskassessmentcontextual, kiyohara2024assessingbenchmarkingriskreturntradeoff, togashi2024safecollaborativefiltering, ge2020learning, gupta2023safe}, but model the random variable at the user, session, or policy levels. Our method complements these approaches by applying the tail-risk logic directly to the diversification step, guaranteeing that even the users with minority intents in the query are satisfied.

\subsection{Diversification in Rankings}
The naive ranking algorithm, which ranks documents in the order of relevance, often only shows the majority intents and ignores the minority intents of a user. The literature of diversification in search has long aimed at providing reasonable satisfaction across intent groups~\citep{agrawal2009diversifysearch, radlanski2008diverseonline, sakai2011diversity, clarke2008alphadiversity}. However, despite such clear connections with the motivation for robustness, most diversification frameworks generally consider the average rather than risks or guarantees. In fact, we find that IA metrics~\citep{agrawal2009diversifysearch} and D-measure~\citep{sakai2011diversity}, two of the most common diversification metrics, provide identical results to a naive, non-diversified metric in many cases. Since many diversification studies are built under their logic, we show empirically that many diversification algorithms perform very poorly in terms of robustness.

Some studies provide diversification methods and metrics aiming to minimize the probability of a user not clicking on any document~\citep{agrawal2009diversifysearch, radlanski2008diverseonline, Wang2025mood}, which is fundamentally similar to our goal of robustness. However, typically assume a cascade behavior assumption \textit{and} associate only click probabilities, which are rare cases in practice. As we will show in the experiments, these methods perform poorly with graded relevance (e.g., nDCG) and explicit ratings. Some other studies consider the coverage of information without explicitly considering intents, thus conceptually different from what we aim~\citep{carbonell1998mmr, Panigrahi2012minmax}.

Fairness in rankings is also a related field that ensures that each candidate group gets a fair amount of exposure in a ranking~\citep{Singh_2018exposurefairness, sakai2022versatileframeworkevaluatingranked}. It does not ensure that the \textit{user} is given a diversity of exposure, thus conceptually different from our goal.

\section{Problem Formulation}
\label{sec:risk-framework}

\subsection{Preliminaries}
\label{sec:problem-setting}
We consider the evaluation of a single top-$\topk$ ranking, denoted $\rankingk$. Let $\query$ denote a query with intent set $\intentsq{\query} = \{\intenti{1}, \ldots, \intenti{\numintents}\}$, where each intent $\intenti{i}$ represents a distinct information need associated with $\query$. 

We denote $\condprob{\intent}{\query}$ as the probability that a user issuing query $\query$ has intent $\intent$. Following common prior work~\citep{agrawal2009diversifysearch, clarke2011noveltydiversity, radlanski2008diverseonline, sakai2011diversity}, we assume the intents are mutually exclusive, i.e., $\sum_{\intent \in \intentsq{\query}} \condprob{\intent}{\query} = 1$.

Let $\relevanceQI{\doc}{\query}{\intent} \in [0,\relmax]$ be the relevance of document $\doc\in\corpus$ for query $\query$ given a user intent $\intent$. We suppose that $\relevanceQI{\doc}{\query}{\intent}$ and $\condprob{\intent}{\query}$ are known or estimated beforehand, following the common setting in prior work on intent-based diversity~\citep{agrawal2009diversifysearch, sakai2011diversity, clarke2008alphadiversity, clarke2009diversity, Wang2025mood, Vargas2011Intent}. Their estimation is outside the scope of this work.

Naturally, the \textit{raw} relevance of the document (that ignores intents) is calculated by the expected relevance
\begin{align}
    \relevanceQ{\doc}{\query} &= \psExpectExt{\intent}{\condprob{\intent}{\query}}{\relevanceQI{\doc}{\query}{\intent}}\label{eq: raw relevance}\\
    &=\sum_{\intent\in \intentsq{\query}}\condprob{\intent}{\query}\relevanceQI{\doc}{\query}{\intent}. \nonumber
\end{align}
This formulation is widely adopted in prior work~\citep{sakai2011diversity, brandt2011dynamic, welch2011diversity, chapelle_intent-based_2011, maistro2021principled, wang2018searchdiversity}.

\subsection{Standard Metric} To measure the value of a ranking, we take the average relevance
\begin{equation}
    \vstd{\rankingk}{\query} = \frac{1}{\topk}\sum_{\doc\in \rankingk}\relevanceQ{\doc}{\query}.
    \label{eq:standard metric}
\end{equation}
While we use average relevance as a default metric throughout the main text, the base metric $\valuemetric$ can be replaced with other standard metrics (e.g., \ndcg, \dcg, \err, \precisionat{\topk}, \rbp), which we will also discuss in our experiments. 

The optimal ranking for Eq.~\eqref{eq:standard metric} and most standard metrics like \ndcg{} would be the naive ranking algorithm that selects documents in the descending order of relevance
\begin{equation}
    \naive{\query} = \argsort_{\doc} \relevanceQ{\doc}{\query}.
    \label{eq: naive}
\end{equation}

However, since the relevance of a document is heavily dependent on the intent probability (see Eq.~\eqref{eq: raw relevance}), such naive rankings may contain only documents of the majority intents, and some minority intents may be completely ignored. This is not robust, as it may result in a poor search session when the user had a minority intent~\citep{steck2018cali, agrawal2009diversifysearch, sakai2011diversity, clarke2008alphadiversity}.

\subsection{Intent-Weighted Metric}
We now present another common evaluation method of a ranking \textit{given some user intents}. We denote the value of a ranking \textit{given intent} $\intent$ as the average relevance
\begin{equation}
    \valuegiven{\rankingk}{\query}{\intent} = \frac{1}{\topk}\sum_{\doc \in \rankingk} \relevanceQI{\doc}{\query}{\intent}.
\end{equation}
Similar to Eq.~\eqref{eq:standard metric}, this denotes the average relevance of a ranking, but for a specific user intent $c$.

Given this per-intent value $\valuegiven{\rankingk}{\query}{\intent}$, another straightforward evaluation method of the ranking would be to compute the intent-weighted average (IW metric), denoted
\begin{equation}
    \viw{\rankingk}{\query} = \sum_{\intent \in \intentsq{\query}} \condprob{\intent}{\query}\, \valuegiven{\rankingk}{\query}{\intent}.
    \label{eq:iw}
\end{equation}
This evaluation method also demonstrates the family of ``IA’’ metrics by \citet{agrawal2009diversifysearch} (e.g., NDCG-IA, MRR-IA), which are diversification metrics aiming to value minority intents more. For example, if the base metric $V$ is replaced with nDCG (i.e., $\text{nDCG}_\text{IW}$), this matches NDCG-IA~\citep{agrawal2009diversifysearch}. This is arguably the most common metric for diversity, as the same intent‑weighted paradigm also underlies widely used diversification metrics and algorithms such as \xquad~\citep{santos2010xquad}, ERR-IA~\citep{chapelle_intent-based_2011}, \dmeasure~\citep{sakai2011diversity}, and \alphandcg~\citep{clarke2008alphadiversity}.

\subsection{IW Metrics are Surprisingly not Robust}
\label{sec:iw_not_robust}

\begin{table}[t]
\centering
\caption{Toy example of how $\ndcg_\text{IW}$ (= NDCG-IA~\citep{agrawal2009diversifysearch}) can neglect diversification in a top-2 ranking. (a) is the setting, and (b) shows the evaluation scores of each ranking. Although the intent probability of $c_1$ is only 0.02 more than $c_2$, $\ndcg_\text{IW}$ favors a non-diversified ranking. As a result, 49\% of the users can be unsatisfied with the ranking.}
\begin{subtable}[t]{0.48\linewidth}
    \centering
    \caption{Setting.}
    \label{tab:intents}
    \begin{tabular}{lcc}
        \toprule
        & $c_1$ & $c_2$ \\
        \midrule
        $\prob{c\mid q}$ & 0.51 & 0.49 \\
        \midrule
        $\relevanceQI{\textcolor{red}{d_1}}{q}{c}$ & 1 & 0 \\
        $\relevanceQI{\textcolor{red}{d_2}}{q}{c}$ & 1 & 0 \\
        $\relevanceQI{\textcolor{blue}{d_3}}{q}{c}$ & 0 & 1 \\
        $\relevanceQI{\textcolor{blue}{d_4}}{q}{c}$ & 0 & 1 \\
        \bottomrule
    \end{tabular}
\end{subtable}
\hfill
\begin{subtable}[t]{0.48\linewidth}
    \centering
    \caption{Evaluation scores.}
    \label{tab:ndcg}
    \begin{tabular}{lcc}
        \toprule
        $R_{k=2}$ & $\text{nDCG}_\text{IW}$ & $\text{nDCG}_\text{std}$ \\
        \midrule
        $[\textcolor{red}{d_1}, \textcolor{red}{d_2}]$ & \textbf{0.600} & \textbf{1.000} \\
        $[\textcolor{red}{d_1},\textcolor{blue}{d_3}]$ & 0.502 & 0.871 \\
        $[\textcolor{blue}{d_3}, \textcolor{blue}{d_4}]$ & 0.400 & 0.667 \\
        \bottomrule
    \end{tabular}
\end{subtable}
\label{tab:combined}
\end{table}
While $\viw{\rankingk}{\query}$ builds the ground of most diversification methods, we find that, surprisingly, when the base metric $\valuemetric$ is linear (e.g., \dcg, average relevance, \precisionat{\topk}), the intent-weighted metric in Eq.~\eqref{eq:iw} is identical to the standard metric in Eq.~\eqref{eq:standard metric}. Formally, using average relevance as the base metric $\valuemetric$, we have
\begin{align}
  \viw{\rankingk}{\query}
  &= \sum_{\intent \in \intentsq{\query}} \condprob{\intent}{\query} \left( \frac{1}{\topk}\sum_{\doc \in \rankingk} \relevanceQI{\doc}{\query}{\intent} \right) \nonumber\\
  &= \frac{1}{\topk} \sum_{\doc \in \rankingk} \left(\sum_{\intent \in \intentsq{\query}} \condprob{\intent}{\query}\,\relevanceQI{\doc}{\query}{\intent}\right) = \vstd{\rankingk}{\query}.
  \label{eq:ia-equals-naive}
\end{align}
This means an algorithm that maximizes $V_\text{IW}$ will favour the majority intent just as the $V_\text{std}$ would. \textbf{Thus, for linear base metrics $\bm{\valuemetric}$, the intent-weighted metric \emph{offers no protection} to minority intents.} Interestingly, \dmeasure~\citep{sakai2011diversity}, another common diversification method, also reduces to the standard metric. Consequently, the spectrum of intent-based diversification metrics and algorithms~\citep{agrawal2009diversifysearch, sakai2011diversity, clarke2008alphadiversity, santos2010xquad} are \textit{ineffective at providing diversity or robustness} with linear base metrics like \dcg, \precisionat{\topk}, and average relevance.

For non-linear base metrics like nDCG, Expected Reciprocal Rank (\err)~\citep{chapelle2009err} and Rank Biased Precision (\rbp)~\citep{moffat2008rbp}, the equality like Eq.~\eqref{eq:ia-equals-naive} no longer holds. Yet, these metrics can still fail to measure the diversity of a ranking. Table~\ref{tab:combined} shows a toy example of how $\text{nDCG}_\text{IW}$ favors a ranking dominated by one intent, despite the two intent probabilities being almost equal. This is not at all robust, as \textbf{49\% of the intents are completely ignored}.

Furthermore, empirical results in Section~\ref{sec: experiments} reveal that intent-weighted diversification algorithms for such non-linear metrics are also barely more robust than naive algorithms.

\section{Risk-Sensitive Metric: VRisk}
\label{sec:risk-metric}
To address the limitation that existing diversification methods fail to account for robustness, we convert diversification to a risk-minimization problem. Specifically, we propose a metric, \textbf{VRisk}, that measures the expected loss of the least-addressed $\beta$-fraction of intents, bringing the \emph{Conditional Value at Risk} (CVaR) concept into IR. VRisk is intuitive, baseline-relative, and easy to tune.

For each possible intent, we define a loss function as the loss of intent-level value against the target level
\begin{equation}
\label{eq:loss}
\loss{\rankingk}{\query}{\intent} = \pospart{\vtgt{\query}{\intent} - \valuegiven{\rankingk}{\query}{\intent}},
\end{equation}
where $\pospart{\cdot}=\max(0,\cdot)$ and $\vtgt{\query}{\intent}$ is a \emph{target level} for performance on intent $\intent$. $\vtgt{\query}{\intent}$ is a baseline function capturing the satisfaction threshold for intent $\intent$. Unless stated otherwise, we use the \textbf{oracle target}, which is the best possible ranking for the intent
\begin{equation}
    \vtgt{\query}{\intent} = \max_{\rankingkpr} \valuegiven{\rankingkpr}{\query}{\intent}.
    \label{eq: oracle}
\end{equation}
By setting loss against the oracle, Eq.~\eqref{eq:loss} answers the question: ``How well does the ranking satisfy the intent compared to the ideal case?'' We use loss instead of the raw value, as it is fair even when there exist more relevant documents for one intent than the others. Note, however, that when the target level is set to the oracle, the minimization of loss matches exactly the maximization of raw value.

\begin{figure}
    \centering
    \includegraphics[width=1.0\linewidth]{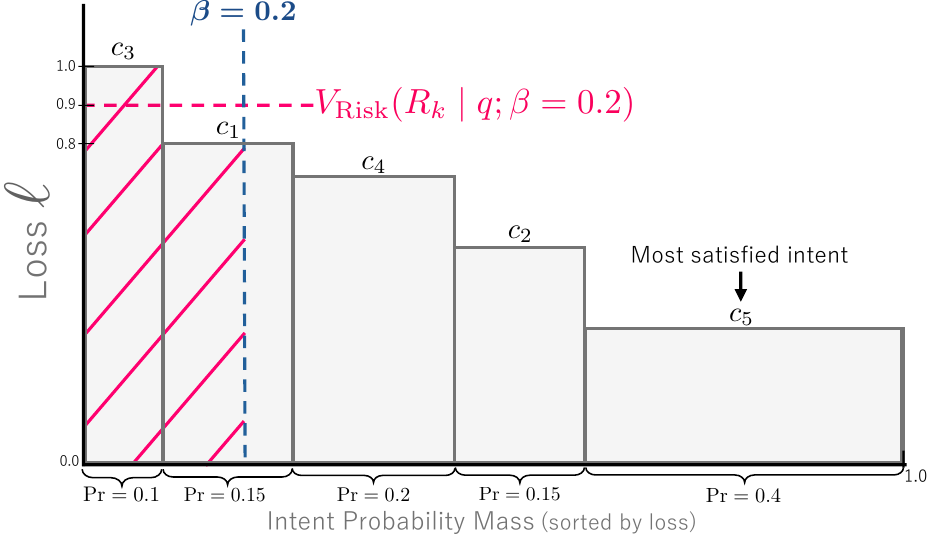}
    \vspace{-10pt}
    \caption{Illustration of VRisk. Bar height shows per-intent loss $\loss{\rankingk}{\query}{\intent}$ and the bar width shows the intent probability $\prob{\intent|\query}$. 5 intents are sorted by loss (worst $\to$ best). VRisk takes the average of \ \area \ , the worst-$\beta$ intent probability mass.}
    \label{fig:vrisk}
    \vspace{-5pt}
\end{figure}

Following previous query-level risk studies~\citep{wang2012risk, dincer2016risk}, $\vtgt{\query}{\intent}$ may also be a value of any baseline algorithm or some heuristic threshold that we want to ensure, which we investigate in Section~\ref{sec: experiments}.

Now, to derive VRisk, let $\zeta$ be the $\beta$‑fraction of the loss distribution, i.e. $\prob{\loss{\rankingk}{\query}{\intent}>\zeta}\le\beta.$ Smaller $\beta$ places more weight on rarer, worse‑served intents. Then, VRisk is the expected loss of the worst $\beta$-fraction of intents, denoted
\begin{equation}
    \vrisk{\rankingk}{\query}{\beta}=\condexpect{\loss{\rankingk}{\query}{\intent}}{\loss{\rankingk}{\query}{\intent}\ge\zeta}.
    \label{eq: cvar raw}
\end{equation}
Smaller VRisk is more robust because it reduces tail loss. Figure~\ref{fig:vrisk} illustrates VRisk when $\beta=0.20$.

Since Eq.~\eqref{eq: cvar raw} is difficult to optimize directly, following the Rockafellar-Uryasev formulation~\citep{Rockafellar2000OptimizationOC}, we express our VRisk metric as
\begin{equation}
\vrisk{\rankingk}{\query}{\beta}=\min_{\zeta\in\mathbb{R}}\left[\zeta+\frac{1}{\beta}\sum_{\intent \in \intentsq{\query}}\condprob{\intent}{\query}\,\pospart{\loss{\rankingk}{\query}{\intent}-\zeta}\right].
\end{equation}
VRisk can tell us both the risk of the ranking and the guarantees, as explained in the example below.
\begin{example}
    If $\, \vrisk{\rankingk}{\query}{\beta}=1.0$ at $\beta=0.05$, we obtain
    \begin{itemize}
        \item Risk: The worst 5\% of intents have an expected loss of 1.0,
        \item Guarantee\footnote{Even better, 95\% of the intents have loss $\leq \zeta\leq 1.0$}: 95\% of the intents have the loss always smaller than 1.0.
    \end{itemize}
\end{example}

\subsection{Property of VRisk}
\label{sec: property}
Our VRisk metric is controllable and generalizable. Specifically, $\beta\in (0, 1]$ is a tunable parameter that controls which worst fractions we would like to care about. An interesting property is described following proposition.
\begin{proposition}
\label{prop:beta1}
When $\beta = 1$, VRisk reduces to the expected loss
\begin{equation}
\vrisk{\rankingk}{\query}{\beta=1} = \sum_{\intent \in \intentsq{\query}} \condprob{\intent}{\query} \cdot \loss{\rankingk}{\query}{\intent}.
\end{equation}
\end{proposition}
Thus, our framework generalizes the average intent-weighted optimization, while enabling control of tail risk via $\beta$. Specifically, a smaller $\beta$ means a \textit{strictly} safer system.

In short, VRisk converts diversification into an intuitive \textbf{risk-minimization} problem driven by a single, intuitive parameter~$\beta$. Using VRisk, we are able to explicitly measure the robustness of a ranking, unlike existing metrics that measure the average.

\paragraph{Why not minimax?}
The minimax criterion, which maximizes the value of the single worst–served case, is typically more common in the IR literature~\citep{Diaz_2024, memarrast2021fairnessrobustlearningrank, Wang_2017irgan, diaz2024recallrobustnesslexicographicevaluation}. In our case, minimax solves
\begin{equation}
\label{eq:minimax}
\min_{\rankingk}\max_{\intent\in\intentsq{\query}} \loss{\rankingk}{\query}{\intent}.
\end{equation}
In Figure~\ref{fig:vrisk}, the minimax task would be to minimize the loss of only the left-most intent (i.e., intent with the most loss).

However, minimax is not suited for our problem, as it is not tunable. Minimax must optimize for the worst intent, even if that intent only occurs at a $0.0001\%$ chance. In reality, we should not lower the average utility just to fulfill such intent. In contrast, VRisk looks at the worst \textit{fraction} of the intents, where $\beta$ controls the size of the fraction, so it addresses a $0.0001\%$ intent \textit{only if} it lowers fraction loss the most. Moreover, \textit{VRisk subsumes minimax} as the special case of $\beta \rightarrow 0$.

\section{Optimization of VRisk: VRisker}
\label{sec:optimization}
The VRisk objective gives us an explicit target of \emph{minimizing tail
risk}, but finding the global optimum is computationally intractable, as shown in the following proposition.
\begin{proposition}[NP-hardness]\label{prop:nphard}
For variable $\topk$ and any $\beta\in(0,1]$, minimizing VRisk
\begin{equation*}
    \min_{\rankingk\subseteq\corpus,\,|\rankingk|=\topk} \; \vrisk{\rankingk}{\query}{\beta}
\end{equation*}
is NP-hard. Proof in Appendix~\ref{app:nphard}.
\end{proposition}

Due to the NP-hardness of the optimization problem, we propose an efficient ranking algorithm called \textbf{VRisk} \textbf{E}fficient \textbf{R}anker (\textbf{\vrisker}), shown in Algorithm~\ref{alg: vrisker}. \vrisker{} is a greedy algorithm that picks documents iteratively from the top position. At each step, it adds the document that most reduces VRisk. When there is a tie, it adds the document that maximizes the IW metric $\viw{R}{\query}$.
\begin{algorithm}[t]
\caption{\vrisker}
\label{alg: vrisker}
\textbf{Require:} Query $\query$, candidates $\corpus$, length $\topk$, risk level $\beta$
\begin{algorithmic}[1]
\State Ranking $R \leftarrow \varnothing$
\For{each rank $i = 1, \dots, \topk$}
    \State \textcolor{gray}{// select document that minimizes VRisk}
    \State $\doc^* \leftarrow \argmin_{\doc \in \corpus \setminus R} \vrisk{R \cup \{\doc\}}{\query}{\beta}$
    \State Tie-break by maximizing $\viw{R \cup \{\doc\}}{\query}$
    \State \textcolor{gray}{// append document to rank $i$}
    \State $R \leftarrow R \cup \{\doc^*\}$
\EndFor
\State \textbf{return} Ranking $R$
\end{algorithmic}
\end{algorithm}

VRisker is practical in terms of complexity. Let $\numintents=|\intentsq{\query}|$ and $\numdocs=|\corpus|$. Each risk evaluation involves $O(\numintents)$ arithmetic operations. The run time is $O(\topk^2 \numdocs \numintents)$, but drops to $O(\topk \numdocs \numintents)$ with incremental updates, which is identical to the speed of prior greedy diversification algorithms~\citep{agrawal2009diversifysearch, chapelle_intent-based_2011, steck2018cali, santos2010xquad}. Empirical comparison is in Figure~\ref{fig:runtime}.

\subsection{Approximation Guarantee of VRisker}
\label{sec:guarantee}
The \vrisker{} algorithm has strong theoretical guarantees. When the base metric $\valuegiven{\cdot}{\query}{\intent}$ is \emph{modular} (e.g., average relevance, \precisionat{\topk}), the per‑intent loss $\loss{R}{\query}{\intent}$ is monotone non‑increasing, and the \emph{risk‑reduction} function
\begin{equation}
    \riskreduction{R}=\vrisk{\varnothing}{\query}{\beta}\;-\;\vrisk{R}{\query}{\beta}
\end{equation}
is \emph{monotone submodular}.\footnote{Submodularity follows because
(i) adding a document can only lower each hinge‑loss term, and
(ii) the convex combination inside CVaR is linear in these losses.} Given the submodularity, we are able to provide an approximation guarantee on \vrisker, in Theorem~\ref{thm:approx}.
\begin{theorem}[$(1-\!1/e)$ Optimality Guarantee]\label{thm:approx}
For modular base metric $\valuemetric$, let $\rankingkopt$ be the optimal length‑$\topk$ ranking.  
Then, \vrisker{} returns $\rankingk$ such that
\begin{equation}
    \riskreduction{\rankingk}\;\ge\;\bigl(1-\tfrac1e\bigr)\,\riskreduction{\rankingkopt},
\end{equation}
i.e., it captures at least $63\%$ of the optimal risk drop. Proof in Appendix~\ref{app:approx}.
\end{theorem}

\subsection{Guarantee for Non-Modular Metrics}
\label{sec:non-modular-guarantee}
When the base metric $\valuemetric$ is non-modular, such as \ndcg, the risk-reduction function $\riskreduction{R}$ is no longer submodular, so Theorem~\ref{thm:approx} does not hold. However, we can establish a formal approximation guarantee for \vrisker{} using the concept of the \emph{submodularity ratio}~\citep{das2011submodular}. A function has a submodularity ratio $\gamma \in (0, 1]$ if the marginal gain of adding an element to a larger set is at least a $\gamma$-fraction of the marginal gain of adding it to a smaller subset.

For example, for VRisk optimization with \ndcg{} as the base metric, we can show that the risk-reduction function $\riskreduction{R}$ has a data-independent submodularity ratio $\gamma > 0$. Building on \citet{bian2019guarantee}, we prove the following theorem.

\begin{theorem}[nDCG-Risk Approximation]\label{thm:ndcg_approx}
For any query $\query$, risk level $\beta$, and cut-off\ $\topk$, VRisker returns a ranking $\rankingk$ that satisfies
\begin{equation}
    \riskreduction{\rankingk}\;\ge\;\bigl(1-e^{-\gamma}\bigr)\,\riskreduction{\rankingkopt},
\end{equation}
where $\gamma \ge {\ndcgdiscount{\topk}}/{\sum_{i=1}^{\topk}\ndcgdiscount{i}}$ is the submodularity ratio, with $\ndcgdiscount{i} = 1/\log_2(1+i)$ being the \ndcg{} discount at rank $i$. Proof in Appendix~\ref{app:ndcg_approx}.
\end{theorem}
Theorem~\ref{thm:ndcg_approx} provides a worst-case guarantee for \vrisker's performance. The ratio $\gamma$ depends only on the ranking depth $\topk$, where for typical depths such as $\topk=5, 10, 20$, the lower bound on $\gamma$ is $0.15, 0.09, 0.05$, respectively. While this theoretical bound is looser than when $\valuemetric$ is modular, we show that \vrisker{} is not an arbitrary heuristic. We also examine the optimality of VRisker on non-modular base metrics in Section~\ref{sec: experiments}.

To sum up, \vrisker{} is a theoretically grounded and efficient algorithm that explicitly minimizes the risks of a user receiving a poor ranking.

\section{Experiments}

\begin{figure*}[t]
    \centering
    \includegraphics[width=1.0\linewidth]{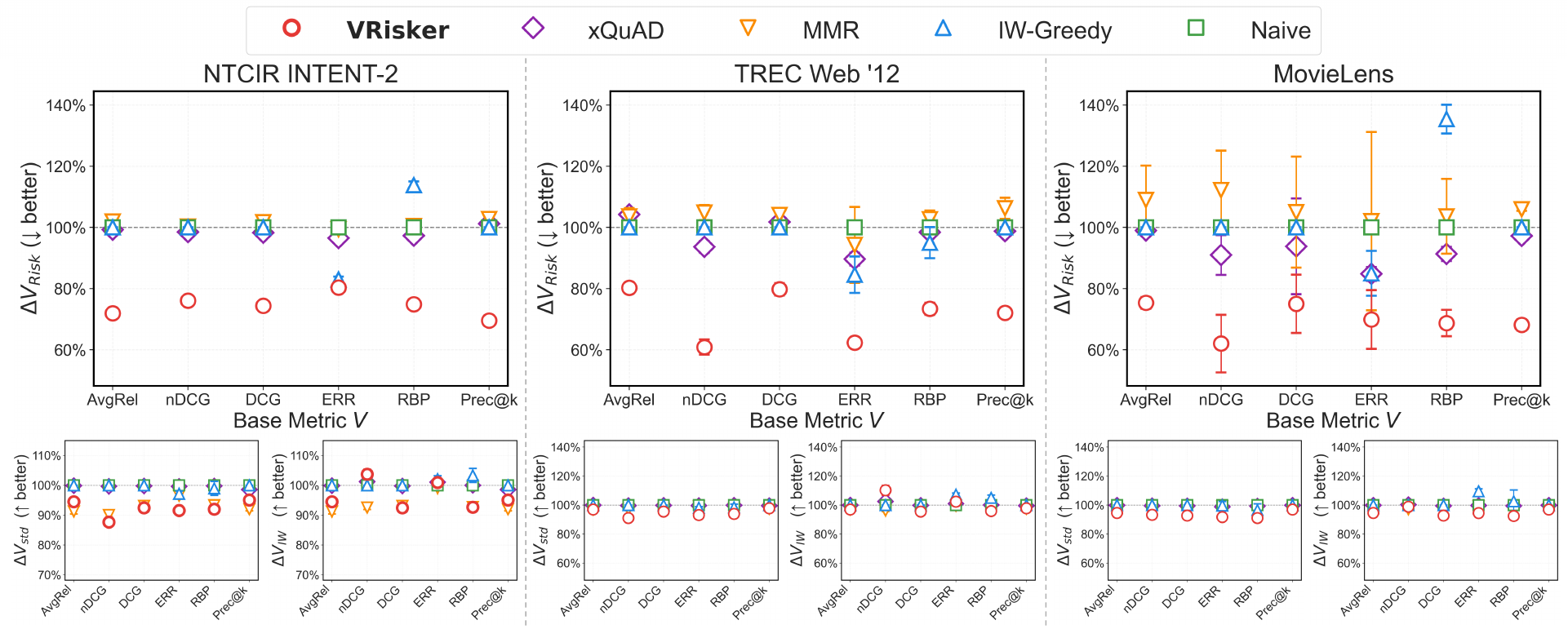}
    \caption{\textit{VRisker is robust across base metrics $V$.} The plots show how different methods perform on varying base metrics $V$, tested on NTCIR INTENT-2 (left), TREC Web 2012 (middle), and MovieLens 32M (right). The top figures evaluate risk, $\Delta V_\text{Risk}$ (smaller is more robust), which measures the expected loss of the worst $\beta$-fraction of intents. The bottom figures evaluate the average performance, $\Delta V_\text{IW}$ and $\Delta V_\text{std}$ (larger is better). All values are relative to Naive = 100\%. ($k=10$, $\beta=0.10$)}
    \label{fig:vary metric}
\end{figure*}

\label{sec: experiments}
In this section, we extensively evaluate our metric and method by comparing with other re-ranking approaches on NTCIR INTENT-2~\citep{sakai2013ntcirintenttask}, TREC Web 2012~\citep{Clarke2012WebTrack}, and on MovieLens 32M datasets~\citep{harper2015movielenslatest}. 


\subsection{Experiment Setup}
\paragraph{Datasets.}
NTCIR INTENT-2 and TREC Web 2012 provide a list of intents for each query, with their intent-specific graded relevance $\relevanceQI{\doc}{\query}{\intent}$. NTCIR INTENT-2 provides 5 relevance grades, and TREC Web 2012 provides 6. While NTCIR INTENT-2 provides their intent probabilities $\condprob{\intent}{\query}$, TREC Web does not provide intent probabilities. However, following \citet{Clarke2012WebTrack}, we let all intents in the TREC Web datasets occur with equal probabilities.

To test generalization beyond search, we also evaluate on MovieLens 32M~\citep{harper2015movielenslatest}, where intent-based diversity is of increasing interest~\citep{Wang2025mood, Vargas2011Intent, Jannach2024Survey, Li2023IntEL, Chang2023Latent, Sun2024LLM, Wang2023IDCL}. We treat each user as a query and each genre as an intent. We select users with more than 200 ratings to ensure sufficient per-user data. Otherwise, different re-ranking algorithms can collapse to similar outputs under extreme sparsity. Following \citet{steck2018cali}, we estimate $\condprob{\intent}{\query}$ from the user’s historical genre proportions and use explicit ratings as raw relevance $\relevanceQ{\doc}{\query}$. For multi-genre items, we define per-intent relevance via a Bayes-consistent allocation
\begin{equation}
    \relevanceQI{\doc}{\query}{\intent}=\mathbbm{1}[c\in C(\doc)]\frac{\relevanceQ{\doc}{\query}}{\sum_{\intentpr\in C(\doc)}\condprob{\intentpr}{\query}},
\end{equation}
where $C(\doc)$ is the set of genre labels for $\doc$ and $\mathbbm{1}$ is an indicator function. This construction ensures Eq.~\eqref{eq: raw relevance}. We adopt an evaluation-only protocol, where no recommender is trained, so as to isolate the effect of the re-ranking objective. Unobserved ratings are treated as non‑relevant for evaluation only.

Additional statistics of the datasets are summarized in Table~\ref{tab: datasets}.

\begin{table}[h]
\centering
\small
\caption{Statistics of the datasets. (Note that the \#Queries denote \#Users in MovieLens.)}
\vspace{-5pt}
\begin{tabular}{lrrr}
\toprule
Dataset & \#Queries & \#Docs & Avg. \#Intents per Query \\
\midrule
INTENT-2 (JP)~\citep{sakai2013ntcirintenttask} & 95  & 5{,}085 & 6.1 \\
TREC Web ’12~\citep{Clarke2012WebTrack}  & 50  & 15{,}200 & 6.0 \\
MovieLens 32M~\citep{harper2015movielenslatest}     & 42{,}902 & 71{,}933 & 8.0 \\
\bottomrule
\end{tabular}
\vspace{-5pt}
\label{tab: datasets}
\end{table}

\paragraph{Experimental Parameters.}
We set $\beta=0.10$ as the default, meaning that we focus on the average loss of the worst $10\%$ of intents when computing $\vriskname$. We also provide experiments where we sweep $\beta$ to different values. The ranking length $\topk$ is also an experimental parameter, with the default length $\topk=10$. The baseline performance $\vtgt{\query}{\intent}$ is also an experimental parameter, but set to the oracle (Eq.~\eqref{eq: oracle}). 

Importantly, we set the base metric $\valuemetric$ to average relevance (which we denote AvgRel) as in the main text. Therefore, for the majority of experiments, we only show results on VRisk and $\vstdname$, as $\vstdname=\viwname$ holds (refer to Section~\ref{sec:iw_not_robust}). However, we also test on \ndcg, \dcg, \err~\citep{chapelle2009err}, \rbp~\citep{moffat2008rbp}, and Precision@k to show the generalizability of VRisk and \vrisker. 

\paragraph{Compared Methods}
We compare \vrisker{} with the following:

\textbf{Naive.} The naive ranking approach naively optimizes the standard metric $\vstdname$ as in Eq.~\eqref{eq: naive}. It does not care about diversity.

 \textbf{IW-Greedy.} \iwgreedy{} is a greedy maximization method that aims to maximize $\viwname$ in Eq.~\eqref{eq:iw}. For linear base metrics (e.g., average relevance, DCG), \iwgreedy{} is \textit{strictly optimal}, exactly matching the Naive method. For non-linear base metrics, IW-Greedy has a $(1-1/e)$ optimality guarantee~\citep{chapelle_intent-based_2011}.

\textbf{xQuAD}~\citep{santos2010xquad}. An intent-based diversification method. $\lambda_\text{xQuAD}$, the weight controls the balance between relevance and diversity, is set to 0.5\footnote{\label{foot:lambda}We sweep this in Appendix~\ref{app: additional} and confirm that it does not change the conclusion.}.

\textbf{MMR}~\citep{carbonell1998mmr}. A diversification method that intends to maximize the coverage without relying on intents. Similarity of documents is computed by measuring the cosine similarity of TF-IDF vectors of document texts for search tasks, and tags and titles for the recommendation tasks. $\lambda_\text{MMR}$, the weight controls the balance between relevance and diversity, is set to 0.5\footref{foot:lambda}.

For all main experiments, we also give comparisons with IA-SELECT~\citep{agrawal2009diversifysearch}, FA*IR~\citep{Zehlike_2017}, and Calibrated Recommendations (CR)~\citep{steck2018cali} in Appendix~\ref{app: additional}.

Unless specified otherwise, instead of reporting the raw metrics, results are shown as a percentage compared to the Naive ranking, so that all curves share the same 100\% center line, making the trade-off between $\vriskname{}$ and average performance (i.e., $\vstdname$ / $\viwname$) more visible. For each $\metricobj\in\{\text{Risk},\text{IW},\text{std}\}$, we set
\begin{equation}
    \deltavobj{\rankingk}=\frac{\vobj{\rankingk}}{\vobj{\naive{\query}}}\times 100.
\end{equation}
For transparency, for each experiment, we also report the results on the raw metrics in Appendix~\ref{app: additional}. Note that $\Delta V_\text{IW}=\Delta V_\text{std}$ for linear bases, so we do not show both in most results.

The values shown are the $\deltavobj{\rankingk}$ averaged on all queries\footnote{Note that we cannot compare with query-level robustness methods~\citep{wang2012risk, dincer2016risk}, as they do not take the average on all queries, so cannot be compared on the same scale.}. The error bars indicate the 95\% confidence interval, calculated by treating each query/user as an independent observation drawn at random from the population. We additionally provide statistical significance testing in Appendix~\ref{app: statistical significance}.

\subsection{Results and Q\&As}
We present our results via the following Q\&As.

\vspace{5pt}
\noindent
\textbf{\underline{Q: Does VRisker work on different base metrics?}}

\noindent
\textbf{A:} VRisker performs robustly across various base metrics, while existing methods are unstable and vulnerable.

Figure~\ref{fig:vary metric} compares the use of different base metrics $V$: AvgRel (average relevance, the default), nDCG, DCG, Expected Reciprocal Rank (ERR)~\citep{chapelle2009err, chapelle_intent-based_2011}, Rank Biased Precision (RBP) on $p=0.8$~\citep{moffat2008rbp}, and Precision@k. For Precision@k, we binarize the relevance labels via threshold $(\text{rel}\text{max}+\text{rel}\text{min})/2$. We compare methods on naive-relative VRisk ($\Delta V_\text{Risk}$) with $\beta=0.10$, standard metric ($\Delta V_\text{std}$), and IW-metric ($\Delta V_\text{IW}$), where Naive is at 100\%.

While VRisker has a $(1-1/e)$ guarantee for modular base metrics (e.g., AvgRel and Precision@$k$), we observe that VRisker is about 20-40\% more robust than naive on other non-modular base metrics as well. Furthermore, VRisker only decreases the standard performance ($\vstdname$) by about 0-10\%. 

Other diversification methods are barely more robust, or even less robust, than the naive baseline. As we have argued in Section~\ref{sec:iw_not_robust}, IW-based diversification methods (i.e., IW-Greedy and xQuAD) are often very similar to the naive baseline, even when the base metrics are not linear. Interestingly, on base metric nDCG, IW-Greedy and the Naive algorithm behave exactly the same on all experimented settings and queries, despite nDCG not being linear (see Eq.~\eqref{eq:ia-equals-naive}). This is because the greedy marginal at each rank reduces to sorting documents by expected gain under intent weights.

\vspace{5pt}
\noindent
\textbf{\underline{Q: How does the risk level, $\beta$, affect the performance?}}

\begin{figure}[h]
    \vspace{-5pt}
    \centering
    \includegraphics[width=1.0\linewidth]{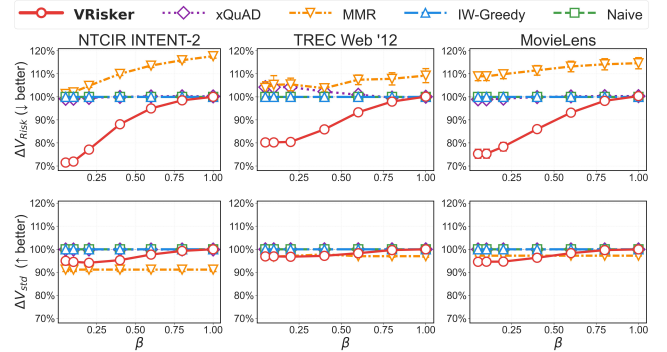}
    \caption{\textit{VRisk/VRisker's pessimism is controllable via $\beta$.} The plots show results on VRisk and $\vstdname$ (= $\viwname$).}
    \label{fig:vary beta}
\end{figure}

\noindent
\textbf{A:} As we make evaluation more pessimistic (smaller $\beta$), average performance decreases, demonstrating the tunability of VRisk \& VRisker.

Figure~\ref{fig:vary beta} shows the performance of different methods when we vary the pessimism, $\beta$. Recall that a smaller $\beta$ focuses on a smaller worst‑case tail of intents. The results demonstrate a clear trade-off between standard performance (i.e., $\vstdname$ / $\viwname$) and VRisk, which is a practical property discussed in Section~\ref{sec: property}. For example, when $\beta=0.05$, VRisker has about 20-30\% less worst-case loss, while sacrificing about 3-5\% in standard performance. In contrast, when $\beta=0.8$, VRisker reduces risk by about 2-3\%, while having lost almost nothing in terms of the standard performance. When the platform should prioritize robustness over standard performance, $\beta$ should be tuned lower.

\vspace{5pt}
\noindent
\textbf{\underline{Q: How does ranking length affect performance?}}
\begin{figure}[h]
    \vspace{-5pt}
    \centering
    \includegraphics[width=1.0\linewidth]{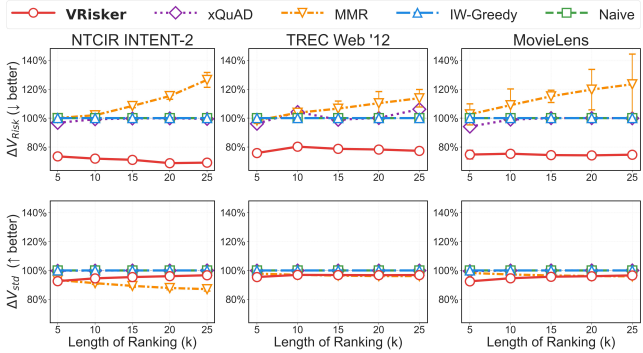}
    \caption{\textit{VRisker is robust to ranking length.}}
    \label{fig:vary k}
\end{figure}

\noindent
\textbf{A:} The average performance of VRisker improves as the ranking length increases, while keeping consistent robustness.

Figure~\ref{fig:vary k} compares the performance on different ranking lengths on the three datasets. VRisker minimizes the worst $\beta$-fraction loss consistently compared to the alternatives, with a 20-30\% decrease. We also observe that in terms of the standard performance ($\vstdname$/$\viwname$), VRisker performs better as the ranking is longer. In the best case, in INTENT-2 $k=25$, we observe $33\%$ reduction in VRisk while sacrificing only $2\%$ in standard performance.

\vspace{5pt}
\noindent
\textbf{\underline{Q: Do we need perfectly accurate intent probabilities?}}

\begin{figure}[h]
    \vspace{-5pt}
    \centering
    \includegraphics[width=1.0\linewidth]{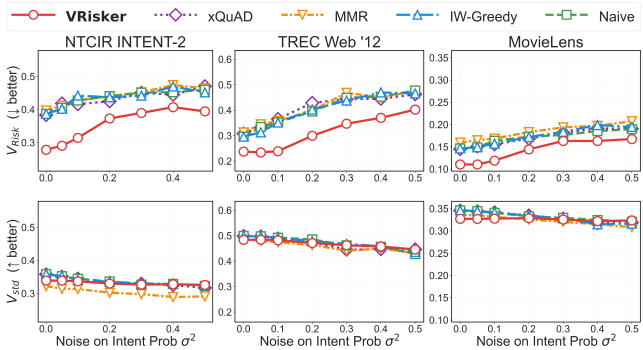}
    \caption{\textit{The results are consistent when intent probabilities with noise added.} VRisk and $\vstdname$ shown with raw values.}
    \label{fig:vary noise}
\end{figure}

\noindent
\textbf{A:} Preferable, but VRisker degrades more gracefully under noise.

Figure~\ref{fig:vary noise} plots the performance of diversification methods on various noise levels. For each intent probability, we perturb each probability as $\Pr(c\mid q)\leftarrow \Pr(c\mid q)+\epsilon_c$, with $\epsilon_c\sim\mathcal N(0,\sigma^2)$. We then clip to $[0,1]$ and renormalize across intents. $\sigma^2$ controls the noise level, where the larger the value, the more noise when running the algorithms. We then evaluate using the true intent probabilities.

We observe that the VRisk increases for all methods, meaning less robustness. This is predictable since noisy intent probabilities trigger noise in both raw relevance calculation and IW metric calculation. Yet, VRisker maintains the lowest risk no matter the noise. Additionally, we observe that in terms of the standard performance, VRisker performs comparatively better as more noise is injected. At $\sigma^2=0.5$, VRisker performs the best on all datasets. This is counterintuitive, but we hypothesize that this is because VRisker satisfies the main intents even with low predicted intent probabilities.

\vspace{5pt}
\noindent
\textbf{\underline{Q: How optimal is VRisker?}}

\begin{figure}[h]
    \vspace{-5pt}
    \centering
    \includegraphics[width=1.0\linewidth]{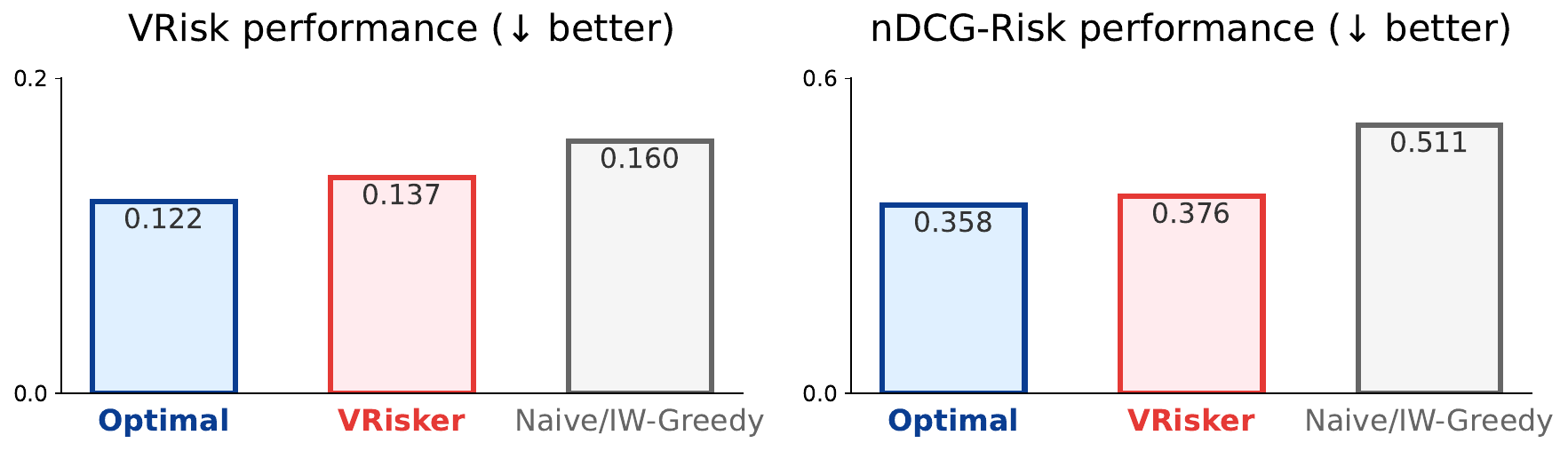}
    \caption{\textit{VRisker is near optimal.} Values are averaged over all three datasets. Results on other bases are in Figure~\ref{fig:app_optim}.}
    \label{fig:optim}
\end{figure}

\noindent
\textbf{A:} VRisker is nearly optimal for various base metrics.

Figure~\ref{fig:optim} shows the comparison of VRisker against the optimal performance, averaged over all three datasets. Here, Optimal is computed exactly per query by solving the VRisk objective as a MILP (PuLP + CBC solver). We show results tested on VRisk (i.e., base metric is average relevance) and nDCG, but we show results on other base metrics in Appendix~\ref{app:optimality}. As discussed in Section~\ref{sec:guarantee}, VRisker is $(1-1/e)$ optimal when the base metric is modular (left-hand chart). As discussed in Section~\ref{sec:non-modular-guarantee}, for non‑modular bases VRisker is not merely a heuristic. As a result, we observe that VRisker is near optimal in all experimented settings, reassuring the robustness of the approach.

\vspace{5pt}
\noindent
\textbf{\underline{Q: How does target level $V_\text{tgt}$ affect the evaluation?}}

\begin{figure}[h]
    \vspace{-5pt}
    \centering
    \includegraphics[width=1.0\linewidth]{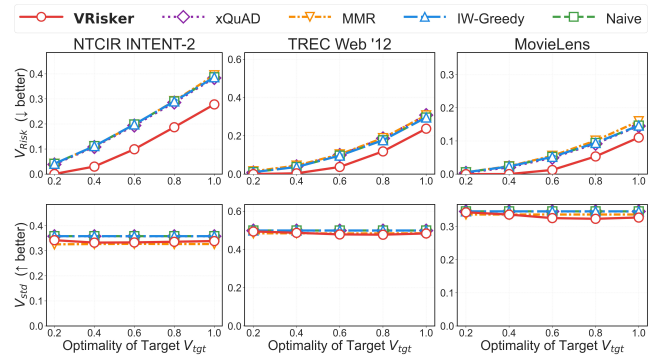}
    \caption{\textit{VRisker works on various baseline target levels, and the optimality of it is a tunable parameter for adjusting safety.} Results shown in raw values.}
    \label{fig:vary tgt}
\end{figure}

\noindent
\textbf{A:} VRisk/VRisker works on various baseline target levels, and target optimality provides a knob to trade off robustness and utility.

In our framework, $V_{\mathrm{tgt}}(q,c)$ can be set to the oracle value or any baseline system. The latter is common in prior IR robustness literature~\citep{wang2012risk, dincer2016risk}, but is now applied within a query.

We test how changing the target level (i.e., baseline system) changes the evaluation and performance. We do this by testing on different optimality of $V_\text{tgt}$, where a smaller value means less optimal. For example, if the optimality is 0.2, this means 20\% of the oracle
\begin{equation}
    \vtgt{\query}{\intent} = 0.2 \times \max_{\rankingkpr} \valuegiven{\rankingkpr}{\query}{\intent}.
\end{equation}
If VRisk=0.0 at target optimality 0.2, this implies that every intent attains at least 20\% of its per‑intent oracle value.

From Figure~\ref{fig:vary tgt}, we observe that VRisk is smaller as the target level is less optimal. This is simply because the target level is more achievable, resulting in lower loss. At optimality 0.2, we observe that VRisker achieves VRisk=0.0, where all intents satisfy at least 20\% of the oracle value. Once perfect VRisk is achieved, VRisker could focus solely on raising the standard performance, because the IW tie-break takes over (see Algorithm~\ref{alg: vrisker}). This can be observed in the bottom figures, where VRisker achieves near-optimal or optimal standard performance at 0.2. Moreover, in the Web'12 and MovieLens datasets, VRisker achieves perfect VRisk and near-optimal standard performance at 0.4 target optimality as well.

This is a clear and interesting trade-off. Lowering the target optimality can assure minimal satisfaction for all intents and also raise the standard performance. On the other hand, higher target optimality makes VRisker strictly safer and makes VRisk more pessimistic. Lastly, another interesting property is that when the target level is the oracle, minimization of loss is exactly the maximization of value. This property is especially relevant in production scenarios where service-level objectives are framed in terms of minimum per-intent quality guarantees. 

By tuning $V_{\mathrm{tgt}}$ (and $\beta$), practitioners can directly express and enforce these guarantees while preserving flexibility for the remainder of the ranking.

\vspace{5pt}
\noindent
\textbf{\underline{Q: How fast is VRisker?}}

\begin{figure}[h]
    \vspace{-5pt}
    \centering
    \includegraphics[width=1.0\linewidth]{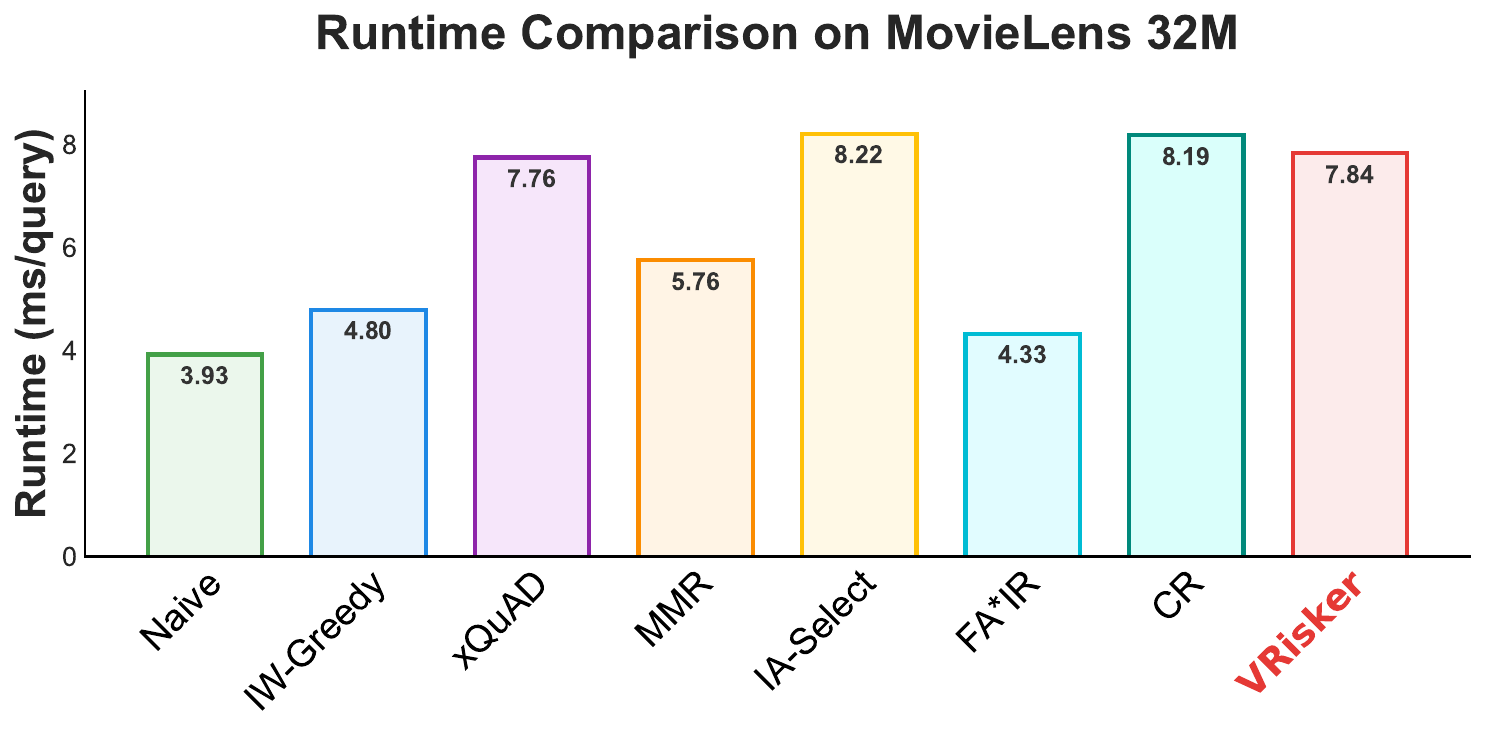}
    \caption{Average runtime on MovieLens 32M (ms/query). Per-query candidate document size is 71,933.}
    \label{fig:runtime}
\end{figure}

\noindent
\textbf{A:} Runtime is comparable to standard diversification methods.

Figure~\ref{fig:runtime} reports per-query runtime\footnote{MacBook Pro (M2, 2022, 16 GB), single-threaded NumPy/BLAS.} to produce a top-$10$ ranking on MovieLens 32M (per‑query candidate set size is $n=71,933$). We compare with three additional baseline methods: IA-SELECT~\citep{agrawal2009diversifysearch}, FA*IR~\citep{Zehlike_2017}, and Calibrated Recommendations (CR)~\citep{steck2018cali} (details of these methods are in Appendix~\ref{app: additional}). 

VRisker achieves 7.84 ms/query, which is within 1\% of xQuAD, and faster than IA-SELECT and Calibrated Recommendation (CR). As expected, greedy diversifiers incur overhead relative to Naive, a simple expected-relevance sort, but the absolute latencies remain very small. These results indicate that VRisker matches the runtime profile of standard greedy diversifiers while delivering its robustness benefits. The complexity of VRisker is shown mathematically in Section~\ref{sec:optimization}.

\vspace{5pt}
\noindent
\textbf{\underline{Q: Do tie-breakers matter?}}

\begin{figure}[h]
    \vspace{-5pt}
    \centering
    \includegraphics[width=1.0\linewidth]{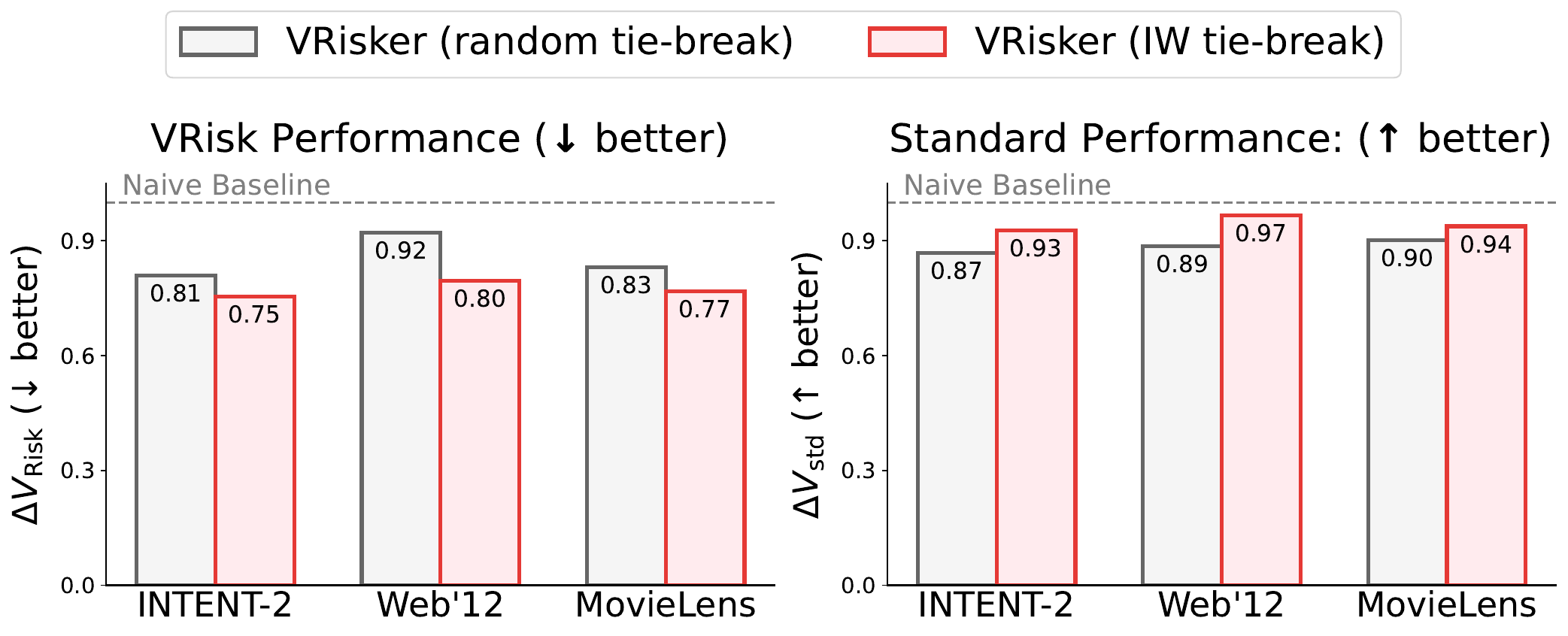}
    \caption{Tie-breaking ablation for VRisker.}
    \label{fig:abl}
\end{figure}

\noindent
\textbf{A:} IW tie-breaker improves both robustness and utility.

When running VRisker, a tie occurs when multiple candidates yield identical VRisk decreases at a position. When there is a tie, VRisker picks the document that maximizes the IW metric (see Algorithm~\ref{alg: vrisker}). Figure~\ref{fig:abl} compares the performance of VRisker with the IW tie‑break versus a random tie‑break. We observe that the IW tie-breaker reduced tail risk by 7–14\% while improving average utility by 4–9\% across INTENT‑2, TREC Web’12, and MovieLens. This is consistent with our objective: once the incremental risk reduction is saturated at a rank, maximizing IW recovers the largest utility without weakening the tail guarantee.

\vspace{5pt}
\noindent
\textbf{\underline{Q: Are existing diversification algorithms robust?}}

\noindent
\textbf{A:} Generally, no.

In most experimental settings (Figures~\ref{fig:vary beta},~\ref{fig:vary k},~\ref{fig:vary noise},~\ref{fig:optim},~\ref{fig:vary tgt}), existing diversification algorithms are no more robust than the naive ranking. This is because the optimization of IW-metrics is identical to the optimization of the standard metrics, as discussed in Section~\ref{sec:iw_not_robust}. In Figure~\ref{fig:vary metric}, we observe that while some (non-linear) base metrics result in less risk, it is not consistent between datasets, and even in the best case, the risk is larger than VRisker. We observe, as \citet{chapelle_intent-based_2011} remark, that IW-Greedy and xQuAD on ERR are fairly robust compared to the other metrics. Thus, if you really need to use an IW-based diversification, we suggest using it with ERR. However, in general, we recommend using VRisker, a tunable, intuitive, baseline-relative, and most importantly robust method.

The appendix provides all results on raw values (i.e., not normalized values) with three additional baseline methods. Figure~\ref{fig:optim} is reported on all base metrics as well. Appendix~\ref{app:lambdasweep} also sweeps the weights of the prior methods to show that the observed results and discussions are not weight-dependent. The code to reproduce the results is provided in \url{https://github.com/RikiyaT/VRisk}. 

\section{Limitations}
\label{app:limitations}
A key limitation, that is shared with prior intent‑aware diversification work~\citep{agrawal2009diversifysearch, sakai2011diversity, clarke2008alphadiversity, Wang2025mood, santos2010xquad}, is the need to estimate intent probabilities $\Pr(c\mid q)$. While recent progress with LLMs makes intent discovery and labeling increasingly tractable~\citep{takehi2025lara, shah2024usinglargelanguagemodels}, such estimates can be biased and typically require calibration against logs or human judgments. Encouragingly, our noise study (Figure~\ref{fig:vary noise}) shows that VRisk remains comparatively robust under perturbed intent distributions.

Our formulation and prior work~\citep{agrawal2009diversifysearch, sakai2011diversity, clarke2008alphadiversity, Wang2025mood, santos2010xquad} also assume mutually exclusive intents per query, which simplifies analysis but may not capture overlapping or hierarchical intents. In addition, VRisk’s behavior depends on two policy parameters: the target level $V_{\mathrm{tgt}}(q,c)$ and the risk level $\beta$. We study both empirically (Fig.~\ref{fig:vary tgt}, Fig.~\ref{fig:vary beta}), but different applications may prefer different settings.

\section{Conclusion}
This paper reframes diversification as \emph{within-query risk minimization}. First, we showed mathematically and empirically that the most common diversification metrics favor majority intents just like standard metrics, prioritizing vulnerable rankings. To address this problem, we introduced \textbf{VRisk}, a CVaR-style, $\beta$-tunable metric that quantifies tail risk. VRisk is intuitive and has various properties that meet practitioner needs. Minimization of VRisk explicitly minimizes the chances of a user failing a search session. To minimize VRisk efficiently, we propose \textbf{VRisker}, a greedy optimizer with a $(1 - 1/e)$ guarantee for modular base metrics and a data-dependent bound for non-modular bases. Empirically, across INTENT-2, TREC Web'12, and MovieLens 32M, VRisker reduced tail risk by up to $33\%$ with only $\sim2\%$ loss in average utility, while classic diversification often matches Naive ranker in robustness. 

For future work, we plan to learn intents jointly, extend to session-level objectives, and integrate the idea of robustness in generative texts like question answering.

\section*{Ethical Considerations}

This work adheres to established ethical standards for research. All evaluations were performed on publicly available datasets containing no personally identifiable information. The proposed methods are intended to enhance user experience by improving diversity and coverage in search and recommendation results.

\bibliography{ref}
\bibliographystyle{ACM-Reference-Format}


\input{app}

\end{document}

%% file: XX-notation.tex
\usepackage{amsmath}
\usepackage{amsfonts}
\usepackage{bm}

\newcommand{\slSet}[1]{\mathcal{#1}}

\newcommand{\query}[0]{q}

\newcommand{\corpus}[0]{\slSet{D}}
\newcommand{\doc}[0]{d}

\newcommand{\intents}[0]{\slSet{C}}
\newcommand{\intentsq}[1]{\intents(#1)}
\newcommand{\intent}[0]{c}
\newcommand{\intenti}[1]{\intent_{#1}}
\newcommand{\numintents}[0]{m}

\newcommand{\condprob}[2]{\text{Pr}(#1 \mid #2)}

\newcommand{\psExpectExt}[3]{\underset{#1 \sim #2}{\mathbb{E}}\left[#3\right]}

\newcommand{\relevanceQI}[3]{\mathrm{rel}(#1 \mid #2, #3)}
\newcommand{\relevanceQ}[2]{\mathrm{rel}(#1 \mid #2)}
\newcommand{\relmax}[0]{\mathrm{rel}_\text{max}}

\newcommand{\topk}[0]{k}
\newcommand{\rankingk}[0]{R_{\topk}}
\newcommand{\naive}[1]{\text{Naive}(#1)}
\newcommand{\argsort}[0]{\operatorname*{argsort}}

\newcommand{\valuemetric}[0]{V}
\newcommand{\valuegiven}[3]{\valuemetric(#1 \mid #2, #3)}

\newcommand{\vstd}[2]{\valuemetric_{\text{std}}(#1 \mid #2)}
\newcommand{\viw}[2]{\valuemetric_{\text{IW}}(#1 \mid #2)}

\newcommand{\dcg}[0]{\text{DCG}}
\newcommand{\ndcg}[0]{\text{nDCG}}
\newcommand{\err}[0]{\text{ERR}}
\newcommand{\rbp}[0]{\text{RBP}}
\newcommand{\precisionat}[1]{\text{Precision@}#1}

\newcommand{\alphandcg}[0]{$\alpha$-nDCG}
\newcommand{\xquad}[0]{\text{xQuAD}}
\newcommand{\dmeasure}[0]{\text{D-measure}}

\newcommand{\rankingkpr}[0]{R'_{\topk}}

\newcommand{\loss}[3]{\ell(#1, #2, #3)}
\newcommand{\vtgt}[2]{\valuemetric_{\text{tgt}}(#1, #2)}
\newcommand{\vrisk}[3]{\valuemetric_{\text{Risk}}(#1 \mid #2; #3)}

\newcommand{\pospart}[1]{\left[#1\right]_{+}}
\newcommand{\prob}[1]{\Pr\bigl(#1\bigr)}
\newcommand{\condexpect}[2]{\mathbb{E}\bigl[#1 \;\big|\; #2 \bigr]}

\newcommand{\vrisker}[0]{\text{VRisker}}

\newcommand{\numdocs}[0]{n}

\newcommand{\riskreduction}[1]{\Phi(#1)}

\newcommand{\rankingkopt}[0]{R_{\topk}^{\star}}

\newcommand{\ndcgdiscount}[1]{w_{#1}}

\newcommand{\intentpr}[0]{c'} 

\newcommand{\vriskname}[0]{\valuemetric_{\text{Risk}}}
\newcommand{\vstdname}[0]{\valuemetric_{\text{std}}}
\newcommand{\viwname}[0]{\valuemetric_{\text{IW}}}

\newcommand{\iwgreedy}[0]{\text{IW-Greedy}}

\newcommand{\metricobj}[0]{\text{obj}}
\newcommand{\deltavobj}[1]{\Delta \valuemetric_{\metricobj}(#1)}
\newcommand{\vobj}[1]{\valuemetric_{\metricobj}(#1)}

%% file: app.tex
\appendix

\section{Proofs}
\label{app: proofs}

\subsection{Proof of Proposition~\ref{prop:nphard}: NP-hardness}
\label{app:nphard}
\begin{proof}
We reduce \emph{Weighted Max-$k$-Cover} to VRisk minimization (with $k$ part of the input).\footnote{For fixed constant $k$, brute force $O(n^k)$ is polynomial.}

Given ground set $U$, weights $w:U\!\to\!\mathbb{R}_{\ge0}$, family $\{S_1,\dots,S_n\}\subseteq 2^U$, and budget $k$, build a single-query instance: intents $\intentsq{\query}=U$ with 
\begin{align*}
    \condprob{u}{\query} &= \frac{w(u)}{W}, \quad W=\sum_{u\in U}w(u).
\end{align*}
Create one document $d_j$ per set $S_j$ and set binary relevance
\begin{align*}
    \relevanceQI{d_j}{\query}{u} = \mathbf{1}[u\in S_j].
\end{align*}
Use the modular base metric and the target $\vtgt{\query}{u}=1/k$; fix $\beta=1$.

For any length-$k$ ranking $R$,
\begin{align*}
    \valuegiven{R}{\query}{u} &= \frac{1}{k}\sum_{d\in R}\mathbf{1}[u\in S(d)], \\
    \loss{R}{\query}{u} &= \max\left\{0,\frac{1}{k}-\valuegiven{R}{\query}{u}\right\} 
    = \frac{1}{k}\,\mathbf{1}\left[u\notin \bigcup_{d\in R}S(d)\right].
\end{align*}
Thus
\begin{align*}
    \vrisk{R}{\query}{1} 
    &= \sum_{u}\condprob{u}{\query}\,\loss{R}{\query}{u} \\
    &= \frac{1}{k}\left(1-\frac{1}{W}\sum_{u\in \cup_{d\in R}S(d)}w(u)\right).
\end{align*}
Minimizing $\vrisk{R}{\query}{1}$ is equivalent to maximizing the covered weight in Weighted Max-$k$-Cover. Hence VRisk minimization is NP-hard.
\end{proof}

\subsection{Proof of Theorem~\ref{thm:approx}: ($1-1/e$)-Optimality Guarantee}
\label{app:approx}

\begin{proof}[Proof of Theorem~\ref{thm:approx}]
Let
\begin{equation*}
\text{VRisk}(R):=\vrisk{R}{\query}{\beta}
=\min_{\zeta\in R}\Bigl[\zeta+\frac{1}{\beta}\sum_i \condprob{\intent_i}{\query}\,(\ell_i(R)-\zeta)_+\Bigr].
\end{equation*}
For the discrete nonnegative losses $\ell_i(R)$, a minimizer $\zeta^\star$ always lies in $[0,\max_i \ell_i(R)]$, hence $\zeta^\star\ge0$. For any fixed $\zeta\ge0$ define
\begin{align*}
&H_\zeta(R)\;:=\;\zeta+\frac{1}{\beta}\sum_i p_i\,\bigl(\ell_i(R)-\zeta\bigr)_+,\\
&\qquad p_i:=\condprob{\intent_i}{\query}.
\end{align*}
Then $\text{VRisk}(R)=\min_\zeta H_\zeta(R)$.

\paragraph{Step 1}
Fix $\zeta\ge0$ and intent $i$. Write $C_i:=\vtgt{\query}{\intent_i}-\zeta\ge0$ and $s_i(R)=\sum_{d\in R}v_i(d)$ with $v_i(d):=\relevanceQI{d}{\query}{\intent_i}/k\ge0$.
For $e\notin R$,
\begin{align*}
\Delta^{(\zeta)}_{i}(R;e)&
:=\bigl(C_i-s_i(R)\bigr)_+ - \bigl(C_i-s_i(R)-v_i(e)\bigr)_+\\
&=\max\{0,\min\{v_i(e),\,C_i-s_i(R)\}\}.
\end{align*}
Hence if $R\subseteq S$ then $s_i(R)\le s_i(S)$ and thus
$\Delta^{(\zeta)}_{i}(R;e)\ge \Delta^{(\zeta)}_{i}(S;e)$.
Summing with nonnegative weights $p_i/\beta$ gives, for all $R\subseteq S$ and $e\notin S$,
\begin{equation*}
H_\zeta(R)-H_\zeta(R\cup\{e\})
\;\ge\;
H_\zeta(S)-H_\zeta(S\cup\{e\}).
\tag{$\ast$}
\end{equation*}
That is, the \emph{risk drop at fixed $\zeta$} has diminishing returns.

\paragraph{Step 2 (one-step greedy progress).}
Let $R_t$ be the greedy set after $t$ steps and let $\zeta_t\in\arg\min_\zeta H_\zeta(R_t)$.
For any $e\notin R_t$,
\begin{equation*}
\text{VRisk}(R_t)-\text{VRisk}(R_t\cup\{e\})
\;\ge\; H_{\zeta_t}(R_t)-H_{\zeta_t}(R_t\cup\{e\})
\end{equation*}
since $\text{VRisk}(R_t)=H_{\zeta_t}(R_t)$ and $\text{VRisk}(R_t\cup\{e\})\le H_{\zeta_t}(R_t\cup\{e\})$.
Let $O$ be an optimal $k$-set. Using ($\ast$) and averaging over $e\in O$,
\begin{align*}
\max_{e\notin R_t}\bigl[H_{\zeta_t}(R_t)-H_{\zeta_t}(R_t\cup\{e\})\bigr]
&\;\ge\;\frac{1}{k}\sum_{e\in O}\bigl[H_{\zeta_t}(R_t)-H_{\zeta_t}(R_t\cup\{e\})\bigr] \\
&\;\ge\;\frac{1}{k}\Bigl(H_{\zeta_t}(R_t)-H_{\zeta_t}(R_t\cup O)\Bigr).
\end{align*}
Monotonicity of $H_{\zeta_t}$ in $R$ gives $H_{\zeta_t}(R_t\cup O)\le H_{\zeta_t}(O)$, and by definition $\text{VRisk}(O)=\min_\zeta H_\zeta(O)\le H_{\zeta_t}(O)$. Therefore
\begin{equation*}
H_{\zeta_t}(R_t)-H_{\zeta_t}(R_t\cup O)
\;\ge\; H_{\zeta_t}(R_t)-H_{\zeta_t}(O)
\;\ge\; \text{VRisk}(R_t)-\text{VRisk}(O).
\end{equation*}
Combining the displays and taking greedy $e_{t+1}$,
\begin{equation*}
\text{VRisk}(R_t)-\text{VRisk}(R_{t+1})
\;\ge\;\frac{1}{k}\bigl(\text{VRisk}(R_t)-\text{VRisk}(O)\bigr).
\end{equation*}

\paragraph{Step 3.}
Let $\Delta_t:=\text{VRisk}(R_t)-\text{VRisk}(O)$. Then
\begin{equation*}
\Delta_{t+1}\le \left(1-\frac{1}{k}\right)\Delta_t,
\end{equation*}
so
\begin{equation*}
\Delta_k\le \left(1-\frac{1}{k}\right)^k\Delta_0\le e^{-1}\Delta_0.
\end{equation*}
Equivalently,
\begin{equation*}
\text{VRisk}(\varnothing)-\text{VRisk}(R_k)
\;\ge\;\left(1-\frac{1}{e}\right)\bigl(\text{VRisk}(\varnothing)-\text{VRisk}(O)\bigr),
\end{equation*}
i.e., the greedy $R_k$ captures at least $(1-1/e)$ of the optimal risk reduction.
\end{proof}

\subsection{Proof of Theorem~\ref{thm:ndcg_approx}: NDCG-Risk Approximation}
\label{app:ndcg_approx}
\begin{proof}[Proof of Theorem~\ref{thm:ndcg_approx}]
Let
\begin{equation*}
F(R) = V_{\text{Risk}}(\varnothing)-V_{\text{Risk}}(R) = \Delta(R)
\end{equation*}
be the \emph{risk--reduction} set function. $F$ is \emph{monotone} because adding a document can only lower $V_{\text{Risk}}$.

\paragraph{Step 1: submodularity ratio.} For any two prefixes $A\subseteq B\subseteq\mathcal D$ and any document $d\notin B$ define the marginal gains
\begin{equation*}
\Delta(d\mid X) = F(X\cup\{d\}) - F(X), \quad X\in\{A,B\}.
\end{equation*}
Write $|A|=r$, $|B|=t$ ($r<t<k$). With the NDCG discount vector $w_1\ge\dots\ge w_k$, $d$ is placed at position $r+1$ in $A\cup\{d\}$ and $t+1$ in $B\cup\{d\}$. For each intent $c$ the per--intent gain satisfies
\begin{align*}
\mathrm{NDCG}_c(A\cup\{d\}) - \mathrm{NDCG}_c(A) &= \frac{w_{r+1}\,g_{d,c}}{\mathrm{IDCG}_c}, \\
\mathrm{NDCG}_c(B\cup\{d\}) - \mathrm{NDCG}_c(B) &= \frac{w_{t+1}\,g_{d,c}}{\mathrm{IDCG}_c},
\end{align*}
where $g_{d,c} = 2^{\mathrm{rel}(d\mid q,c)} - 1 \ge 0$. Because $V_{\text{Risk}}$ is a positive convex combination of these per--intent gains clipped by a hinge at $\zeta$, the clipping can \emph{only reduce} both numerators \emph{by the same amount}. Hence
\begin{align*}
\Delta(d\mid B) &\ge \frac{w_{t+1}}{w_{r+1}}\,\Delta(d\mid A) \\
&\ge \frac{w_k}{w_1}\,\Delta(d\mid A) \\
&\ge \underbrace{\frac{w_k}{\sum_{i=1}^{k}w_i}}_{=\gamma}\,\Delta(d\mid A).
\end{align*}
Thus the \emph{submodularity ratio} \cite{das2011submodular} of $F$ is lower-bounded by $\gamma$.

\paragraph{Step 2: greedy approximation.} For any monotone set function with submodularity ratio $\gamma$, the standard greedy algorithm attains the guarantee
\begin{equation*}
F(R_k)\ge(1-e^{-\gamma})F(R_k^\star)
\end{equation*}
under a cardinality constraint $k$. Applying this to $F=\Delta$ proves Theorem~\ref{thm:ndcg_approx}.
\end{proof}

\section{Additional Experiments and Results}
\subsection{Statistical Significance Testing}
\label{app: statistical significance}
Following IR best practice, we treat the set of queries in each benchmark as a sample from a larger population and test the null hypothesis that two systems have equal \emph{expected} performance across that population. For every query $q$ we compute the paired difference $d_{q}=M_A(q)-M_B(q)$ where $M$ is either our risk metric $V_\text{Risk}(R_k|q;\beta = 0.10)$ or the underlying expected-utility metric $V_{\text{std}}$. We then apply (i) a two-sided Wilcoxon signed–rank test and (ii) a paired randomization test with $B=100{,}000$ permutations. 
\begin{table}[h]
  \centering
  \small
  \caption{Paired significance tests at $\beta=0.10,\;k=10$. Asterisks mark values that remain below $\alpha=0.05$ after Holm–Bonferroni correction over all comparisons.}
  \vspace{-5pt}
  \setlength{\tabcolsep}{4pt}
  \begin{tabular}{lccccc}
    \toprule
    \multirow{2}{*}{Dataset} & \multirow{2}{*}{Metric} &
      \multicolumn{2}{c}{Wilcoxon $p$} &
      \multicolumn{2}{c}{Perm.\ $p$} \\
    \cmidrule(lr){3-4}\cmidrule(lr){5-6}
     & & vs Naive & vs xQuAD & vs Naive & vs xQuAD \\
    \midrule
    \multirow{2}{*}{INTENT-2}
      & $V_\text{Risk}$ & $1.1\times10^{-13}$\textsuperscript{*} & $1.8\times10^{-13}$\textsuperscript{*} & $<10^{-5}$\textsuperscript{*} & $<10^{-5}$\textsuperscript{*} \\
      & $V_{\text{std}}$& $1.2\times10^{-11}$\textsuperscript{*} & $1.2\times10^{-11}$\textsuperscript{*} & $<10^{-5}$\textsuperscript{*} & $<10^{-5}$\textsuperscript{*} \\ 
    \addlinespace
    \multirow{2}{*}{WEB}
      & $V_\text{Risk}$ & $7.7\times10^{-6}$\textsuperscript{*}  & $1.4\times10^{-6}$\textsuperscript{*} & $<10^{-5}$\textsuperscript{*} & $1.0\times10^{-5}$\textsuperscript{*} \\
      & $V_{\text{std}}$& $5.8\times10^{-9}$\textsuperscript{*}  & $5.8\times10^{-9}$\textsuperscript{*} & $<10^{-5}$\textsuperscript{*} & $<10^{-5}$\textsuperscript{*} \\
      \addlinespace
      \multirow{2}{*}{ML 32M}
      & $V_\text{Risk}$ & $7.5\times10^{-65}$\textsuperscript{*} & $9.6\times10^{-65}$\textsuperscript{*} & $<10^{-5}$\textsuperscript{*} & $<10^{-5}$\textsuperscript{*} \\
      & $V_{\text{std}}$& $2.6\times10^{-10}$\textsuperscript{*} & $1.4\times10^{-1}$        & $<10^{-5}$\textsuperscript{*} & $<10^{-5}$\textsuperscript{*} \\
    \bottomrule
  \end{tabular}
  \label{tab:sig}
  \vspace{-5pt}
\end{table}

Table~\ref{tab:sig} reports the resulting $p$-values. Asterisks mark values that remain below $\alpha=0.05$ after Holm–Bonferroni correction over all comparisons.

\subsection{Additional Results on Optimality}
\label{app:optimality}
\begin{figure*}[t]
    \centering
    \includegraphics[width=0.9\linewidth]{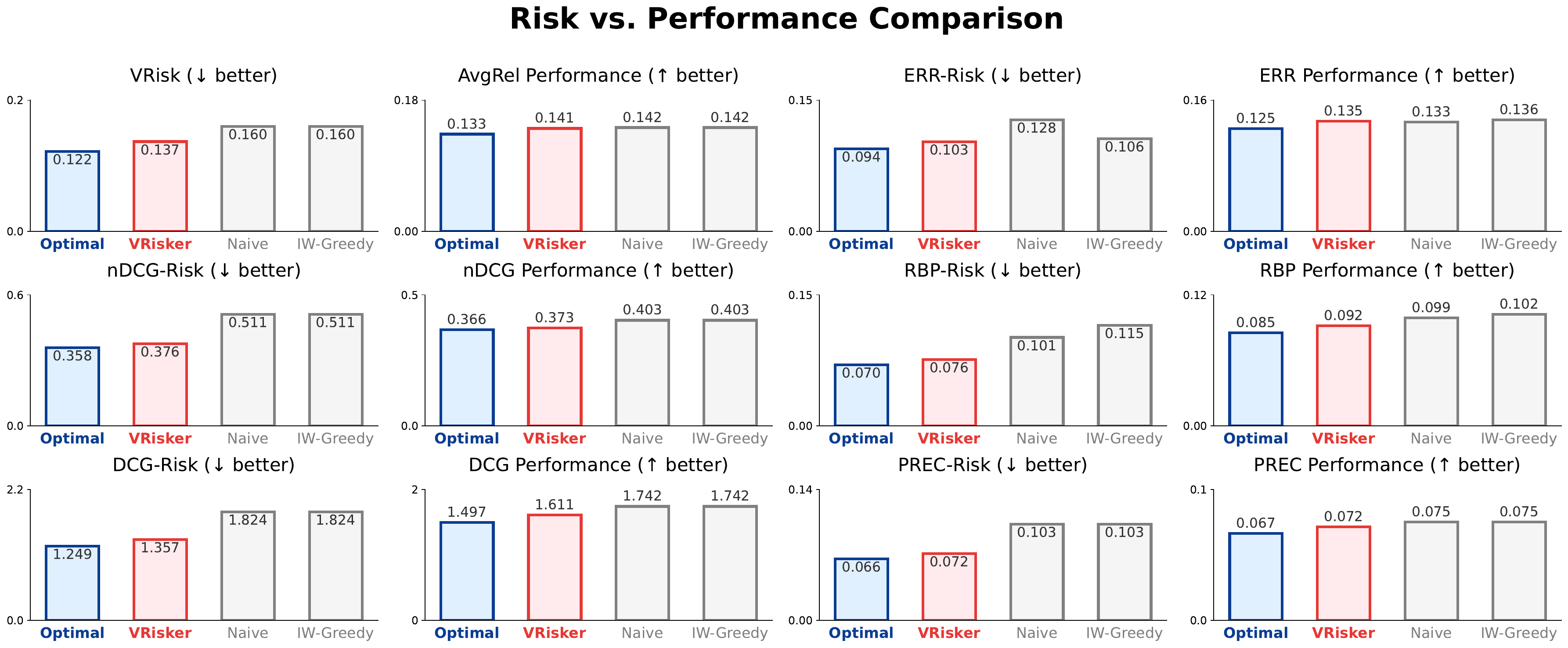}
    \caption{Optimality of VRisker. For each base metric, the left shows VRisk (lower is better) and the right shows $V_\text{IW}$ (higher is better). Figure~\ref{fig:optim} of the main text.}
    \label{fig:app_optim}
\end{figure*}
Figure~\ref{fig:app_optim} provides the full list of experiments on the optimality of VRisker (extended version of Figure~\ref{fig:optim} in the main text). We observe that VRisker is nearly optimal on all metrics. We also observe that the optimal performance, in return, lowers the standard performance on all metrics. This demonstrates that VRisker is robust while being strong in standard performance.

\subsection{Sensitivity Analysis on Prior Methods}
\label{app:lambdasweep}
\begin{figure*}[t]
  \centering
  \includegraphics[width=0.9\linewidth]{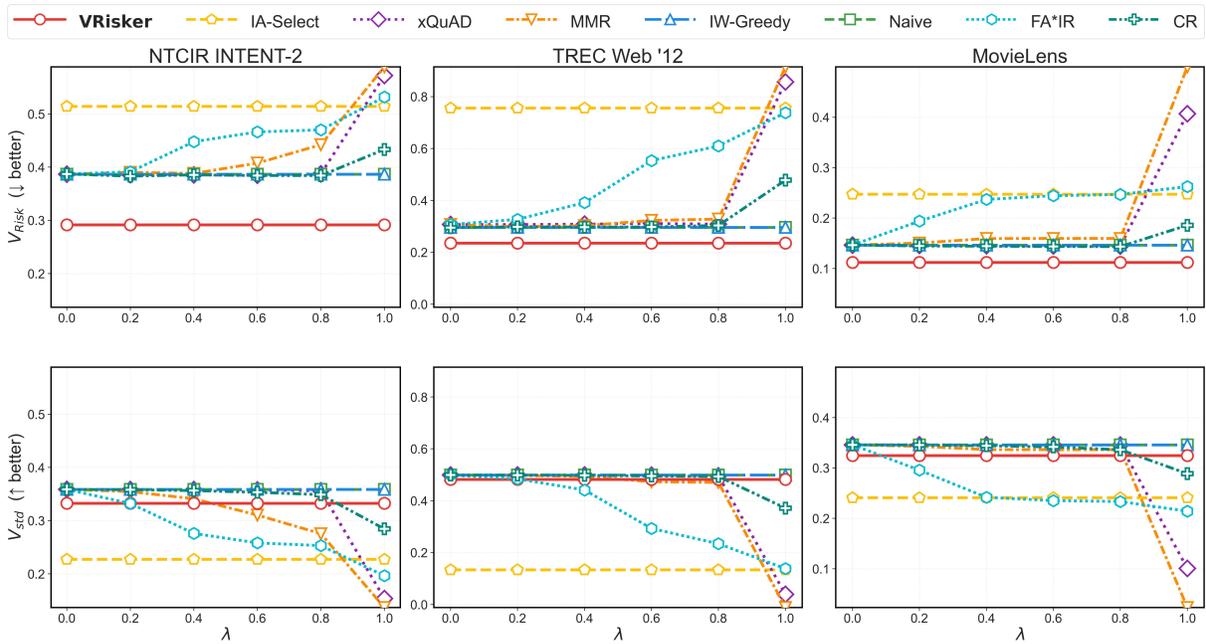}
  \caption{\textbf{$\lambda$-sensitivity (xQuAD, MMR, FA*IR, CR) on three datasets.} As $\lambda \!\to\! 1$ (pure diversity/calibration), all classical baselines \emph{collapse} in $V_{\text{std}}$ and do not provide meaningful tail-risk reduction; for intermediate $\lambda$ they behave similarly to Naive. In contrast, VRisker (no $\lambda$) remains consistently robust with small utility loss.}
  \label{fig:lambda-sweep-all}
\end{figure*}
xQuAD, MMR, FA*IR, and CR all have a weight that balances between relevance and diversity/fairness/calibration. Although they have different notations for weights, we generalize as $\lambda_\text{method}$ in our paper. Essentially, using the weight, they optimize below
\begin{equation*}
\label{eq:lambda-generic}
(1-\lambda_\text{method})V_\text{std}(d \mid q, R) \;+\; \lambda_\text{method}V_\text{X}(d \mid q, R).
\end{equation*}
Thus $\lambda{=}0$ reduces to pure relevance (Naive/IW), and $\lambda{=}1$ to the pure diversity/fairness/calibration.

Figure~\ref{fig:lambda-sweep-all} shows the sensitivity experiment on the same settings as the main text.

Across NTCIR INTENT-2, TREC Web'12, and MovieLens, sweeping $\lambda$ confirms a consistent pattern: (i) at $\lambda{=}1$ all classical baselines become diversity/calibration-only and suffer catastrophic drops in both $V_{\text{std}}$ (=$V_{\text{IW}}$) and $V_{\text{Risk}}$, (ii) for $\lambda \in [0,0.8]$ they track Naive on both risk and utility, offering little worst-case protection, and (iii) VRisker dominates on $V_{\text{Risk},k}$ while maintaining high $V_{\text{std}}$.

This justifies the fixed $\lambda$ choice in the main text and shows that the conclusions are not an artifact of particular hyperparameters.

\subsection{Additional Results}
\label{app: additional}
 In this appendix, we provide experiment results on the main text but with more baselines and with the raw metric values if they are normalized on the main text.

The additional baselines are:

\textbf{IA-SELECT~\citep{agrawal2009diversifysearch}.} This method builds a ranking that minimizes the chances that users will not click on any item in a ranking. Since IA-SELECT works best on a cascade assumption with click probabilities instead of relevance, we normalize the relevance scores to $[0, 1]$ when computing.

\textbf{Calibrated Recommendation (CR)~\citep{steck2018cali}.} This method aims to match the intent-probability proportions to the per-intent utility proportions in the ranking. $\lambda_\text{CR}$ is the weight which controls the balance between relevance and calibration, and is set to 0.5.

\textbf{FA*IR~\citep{Zehlike_2017}.} This method aims to fairly expose each group of documents in the ranking. In our case, we target intent-probability distribution as a fairness objective, trading off between relevance and fairness via the weight $\lambda_{\text{FA*IR}}$ (set to 0.5).

Figures~\ref{fig:app_metric},~\ref{fig:app_beta},~\ref{fig:app_k},~\ref{fig:app_c}, and ~\ref{fig:app_tgt} report the additional experimental results Figures~\ref{fig:vary metric},~\ref{fig:vary beta},~\ref{fig:vary k},~\ref{fig:vary noise}, and~\ref{fig:vary tgt} in the main text respectively.

We observe that the additional baselines are not robust or equal as robust as the naive baseline. This is mainly because the settings and the objectives of the methods are conceptually different from our objective. IA-SELECT only functions under the use of click probabilities under the cascade user assumption. FA*IR and CR respectively focus on fairness and calibration, which are different from diversity and robustness.

\begin{figure*}
    \centering
    \includegraphics[width=1.0\linewidth]{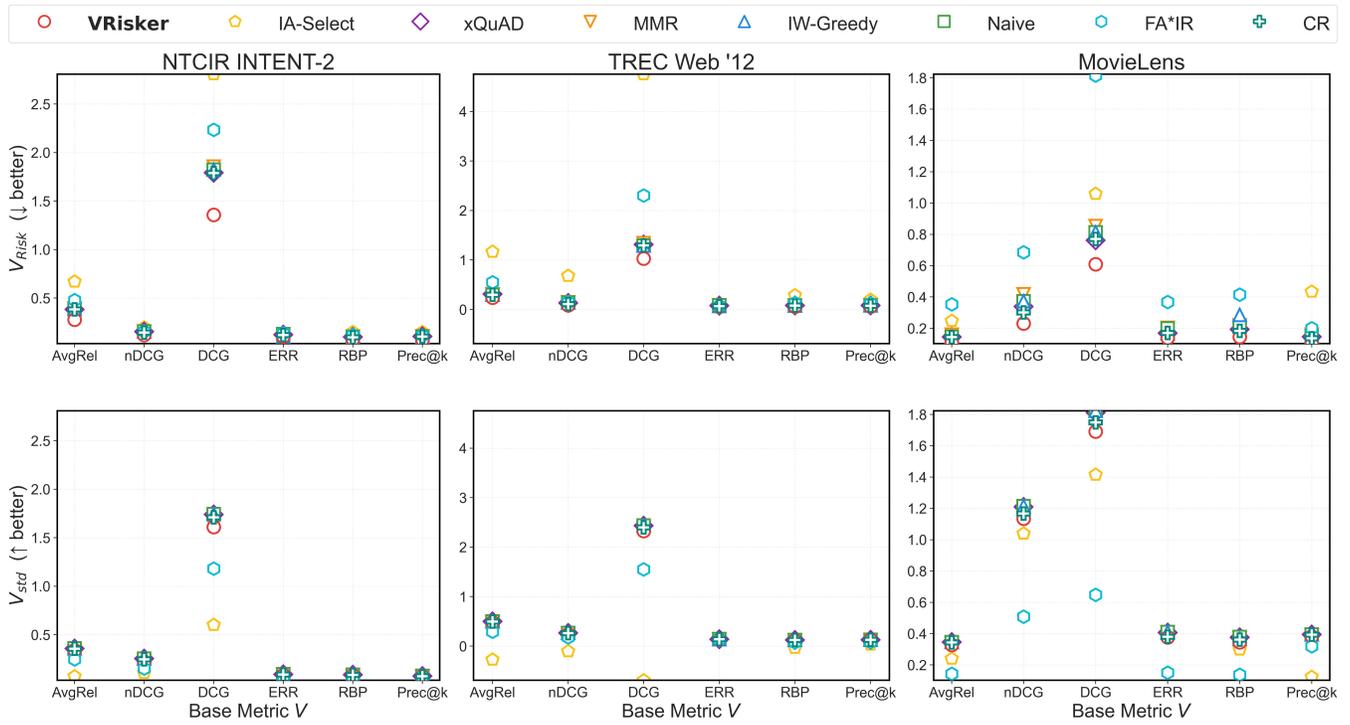}
    \caption{Risk and utility across base metrics. $\Delta V_{\text{Risk}}$ (lower is better) and $\Delta V_{\text{std}} / \Delta V_{\text{IW}}$ (higher is better) for AvgRel, nDCG, DCG, ERR, RBP, and Prec@k on INTENT-2, WEB’12, and MovieLens; Naive = 100\%. Figure~\ref{fig:vary metric} of the main text.}
    \label{fig:app_metric}
\end{figure*}

\begin{figure*}
    \centering
    \includegraphics[width=1.0\linewidth]{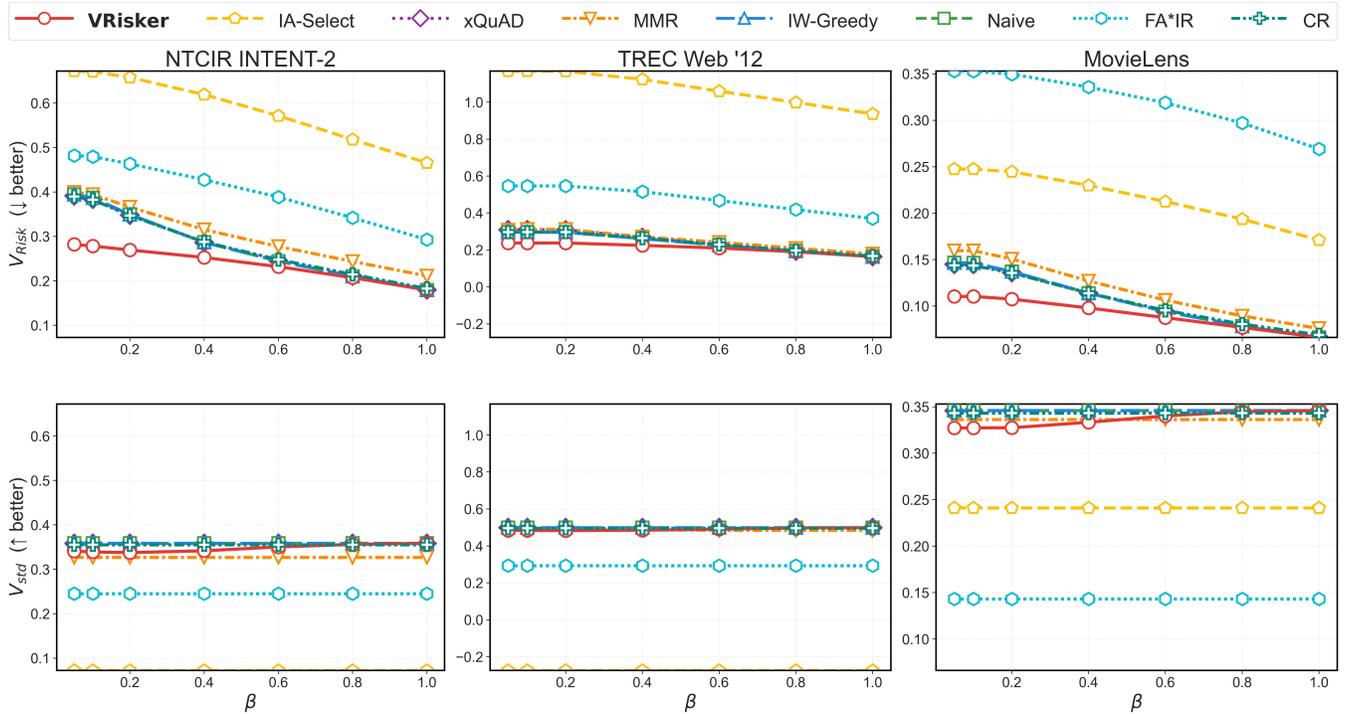}
    \caption{Effect of $\beta$ (pessimism level). As $\beta$ decreases, VRisker reduces tail loss more aggressively with a manageable drop in average utility. Figure~\ref{fig:vary beta} of the main text.}
    \label{fig:app_beta}
\end{figure*}

\begin{figure*}
    \centering
    \includegraphics[width=1.0\linewidth]{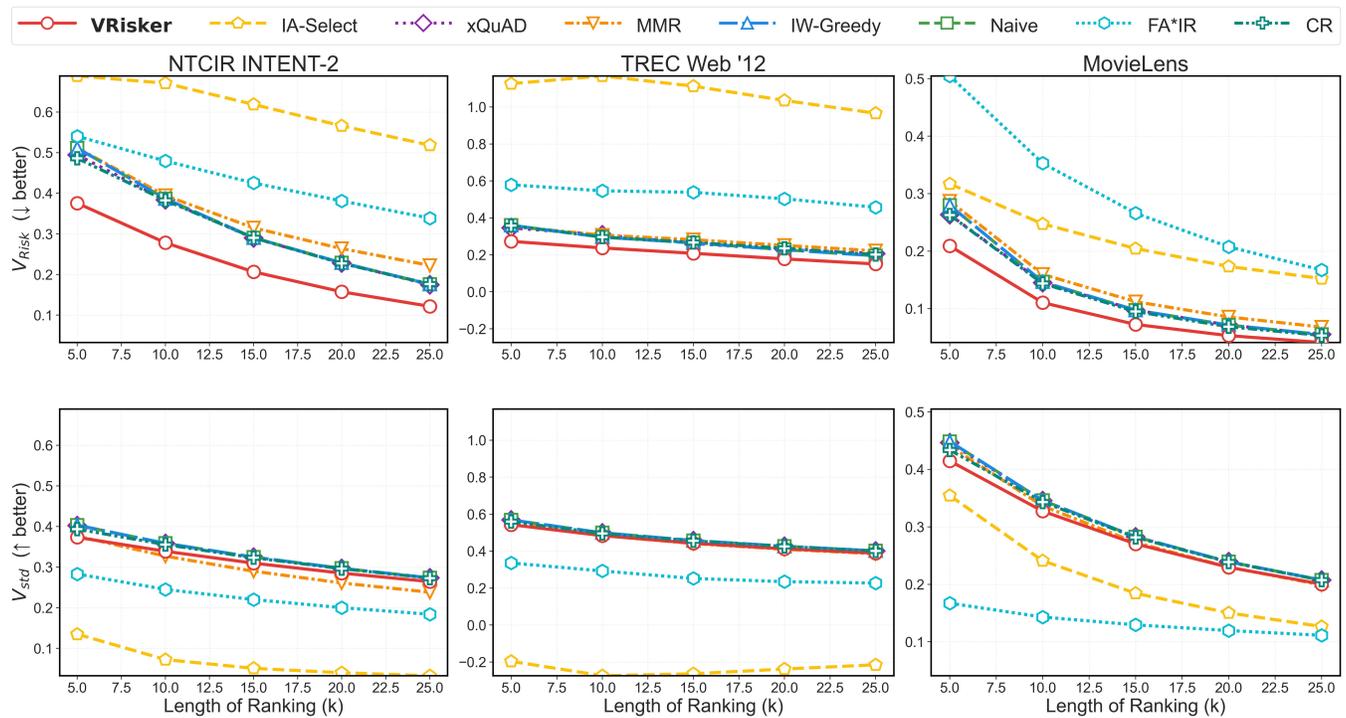}
    \caption{Effect of ranking length $k$. VRisker’s utility improves as $k$ grows while maintaining substantially lower tail risk than baselines. Figure~\ref{fig:vary k} of the main text.}
    \label{fig:app_k}
\end{figure*}

\begin{figure*}
    \centering
    \includegraphics[width=1.0\linewidth]{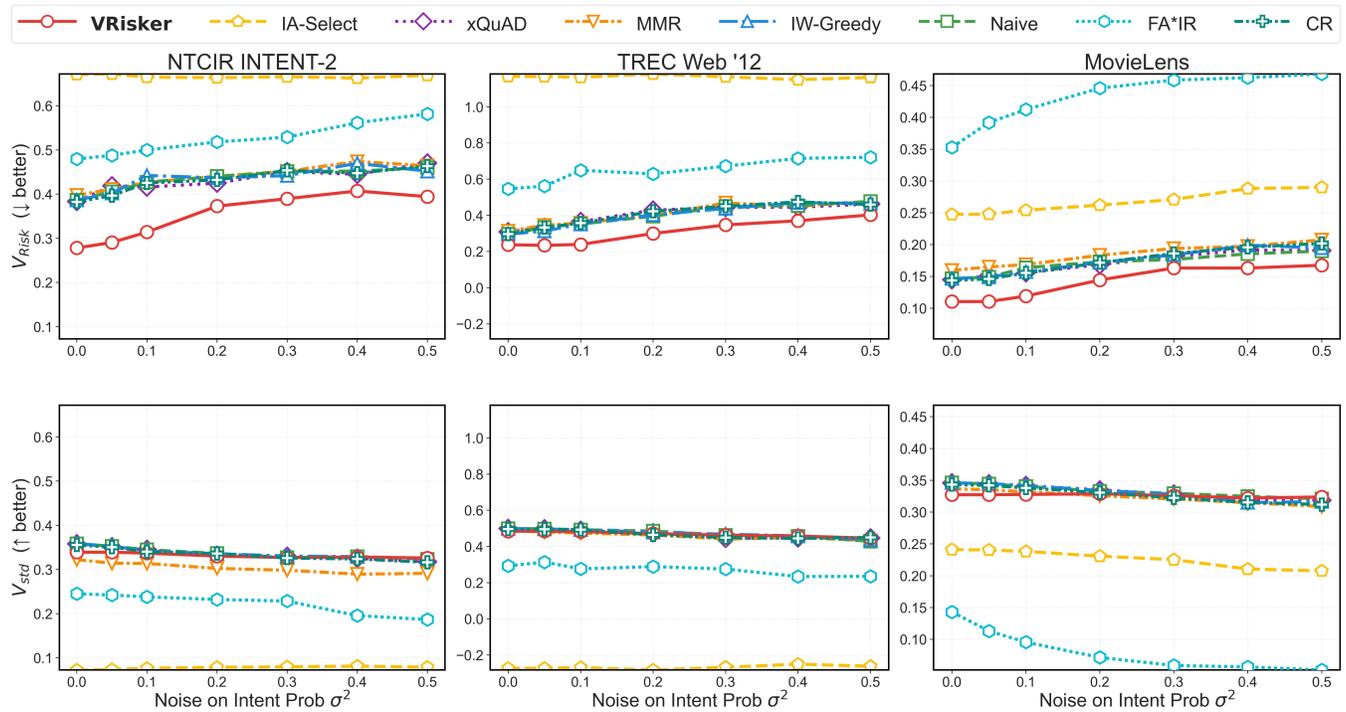}
    \caption{Noise in intent probabilities. VRisker remains most robust under Gaussian noise added to $Pr(c|q)$. Figure~\ref{fig:vary noise} of the main text.}
    \label{fig:app_c}
\end{figure*}

\begin{figure*}
    \centering
    \includegraphics[width=1.0\linewidth]{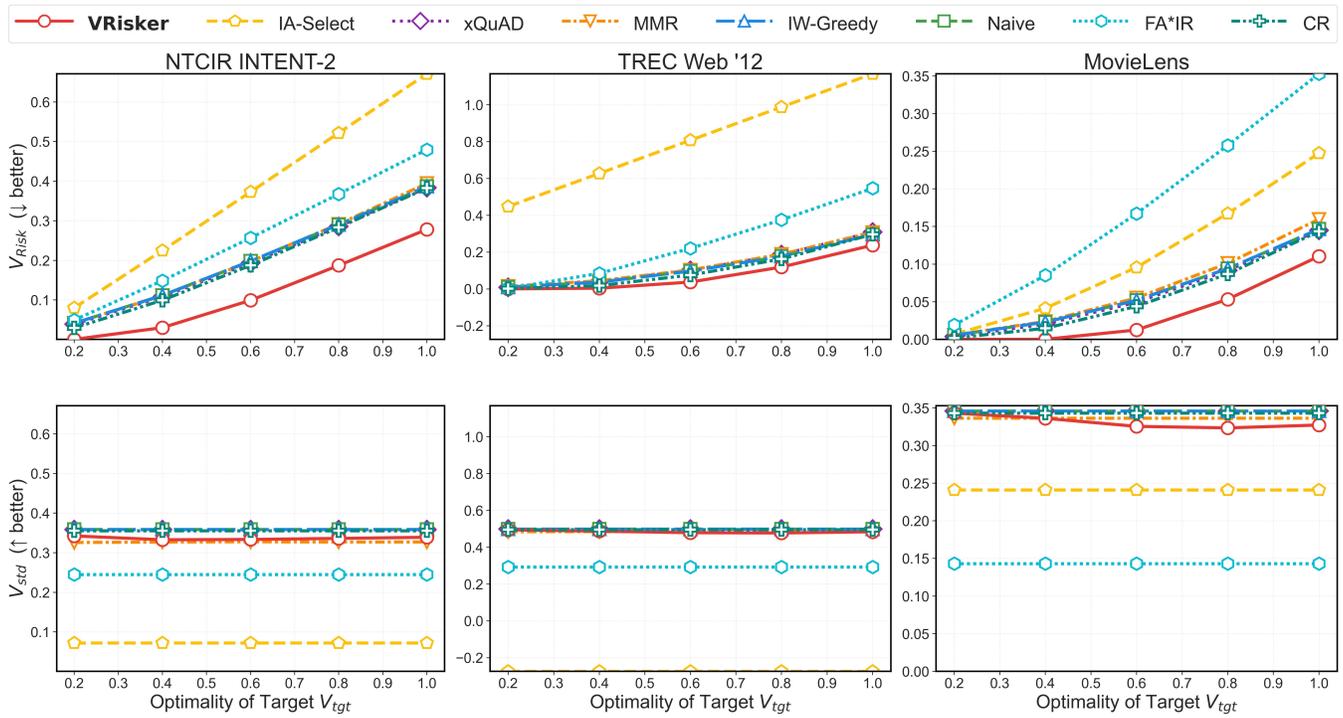}
    \caption{Target optimality sweep. Lower target optimality yields smaller losses and lets VRisker focus on utility once $V_{\text{Risk}} \approx 0$. Figure~\ref{fig:vary tgt} of the main text.}
    \label{fig:app_tgt}
\end{figure*}